\documentclass[11pt,fullpage]{article}

\RequirePackage{graphicx}
\RequirePackage{amsmath} 
\RequirePackage{amssymb}  
\RequirePackage{amsthm}

\RequirePackage{mathrsfs} 
\RequirePackage{txfonts} 

\RequirePackage[numbers]{natbib}
\RequirePackage{tikz}
\usetikzlibrary{matrix,arrows}
\usetikzlibrary{matrix,arrows,shapes.geometric}
\RequirePackage{enumerate}

\RequirePackage{appendix} 

\RequirePackage{color}  

\RequirePackage{multirow}
\RequirePackage{rotating}

\RequirePackage{authblk}

\renewcommand{\cite}[1]{\citep{#1}}

\newcommand{\showfigure}[1]{{#1}}

\newcommand{\ShowFigures}{\renewcommand{\showfigure}[1]{{##1}}}

\ShowFigures

\newcommand{\fullversion}[1]{{#1}}
\newcommand{\shortversion}[1]{}

\newcommand{\ShortVersion}{%
  \renewcommand{\fullversion}[1]{}%
  \renewcommand{\shortversion}[1]{{##1}}}

\ShortVersion

\newcommand{\E}{{\mathbb{E}}}
\newcommand{\R}{{\mathbb{R}}}

\newcommand{\set}[1]{\mathcal{#1}}
\renewcommand{\vec}[1]{\boldsymbol#1}
\newcommand{\equ}[1]{eq.~(\ref{eq:#1})}
\newcommand{\fig}[1]{Fig.~\ref{fig:#1}}

\newcommand{\cart}{\prod} 

\newcommand{\grad}{\mathtt{grad}}

\newcommand{\bftheta}{{\boldsymbol\theta}}

\newtheorem{theorem}{Theorem}
\newtheorem*{theorem*}{Theorem}

\newtheorem{corollary}[theorem]{Corollary}

\newtheorem{definition}{Definition}
\newtheorem{example}{Example}

\newtheorem{lemma}[theorem]{Lemma}

\newtheorem{proposition}[theorem]{Proposition}

\newcommand{\be}{\begin{equation}}
\newcommand{\bel}[1]{\begin{equation}\label{#1}}
\newcommand{\qe}{\end{equation}}
\newcommand{\ee}{\end{equation}}
\newcommand{\eeq}{\end{equation}}
\newcommand{\ba}{\begin{eqnarray}}
\newcommand{\ea}{\end{eqnarray}}
\newcommand{\rf}[1]{(\ref{#1})}

\date{\today}                      

\title{Value of information in noncooperative games}

\author[1]{Nils Bertschinger\thanks{nils.bertschinger@mis.mpg.de}}
\author[2]{David H. Wolpert\thanks{http://davidwolpert.weebly.com}}
\author[1]{Eckehard Olbrich\thanks{olbrich@mis.mpg.de}}
\author[1,2]{J\"urgen Jost\thanks{jjost@mis.mpg.de}}
\affil[1]{Max Planck Institute for Mathematics in the Sciences, Inselstra\ss e 22, D-04103 Leipzig, Germany}
\affil[2]{Santa Fe Institute, 1399 Hyde Park Road, Santa Fe, NM 87501, USA}

\begin{document}

\maketitle
 
\begin{abstract}
 
In some games, additional information  
hurts a player, e.g., in games with first-mover advantage,
the second-mover is hurt by seeing the first-mover's move.   
What properties of a game determine whether it has such negative ``value of
information'' for a particular player? Can a game have negative value of information for all players? 
To answer such questions, we generalize the definition of marginal
utility of a good 
to define the marginal utility of
a parameter vector specifying a game.
So rather than analyze the global structure of the relationship between a game's parameter
vector and player behavior, as in previous work, we focus on the local structure of
that relationship. This allows us to prove that generically, \underline{every} game
can have negative marginal value of information, unless one imposes \emph{a priori}
constraints on allowed changes to the game's parameter vector.  We
demonstrate these and related results numerically, and discuss their implications.
\end{abstract}

\section{Introduction}

How a player in a noncooperative game behaves typically
depends on what information she has both about her physical environment and about the 
behavior of the other players. Accordingly, the joint behavior of multiple interacting
players can depend strongly on the information available to the
separate players, both about one another, and about Nature-based random
variables. Precisely \emph{how} the joint behavior of the players depends on this information
is determined by the preferences of those
players. So in general there is a strong interplay among the information
structure connecting a set of players, the
preferences of those players, and their behavior. 

Previous analyses of this interplay have considered how player
behavior changes under global, non-infinitesimal changes to the parameters specifying the 
underlying game. Here we pursue a different approach, 
generalizing the concept of ``marginal value of a good" from the setting
of a single decision-maker in a game against Nature to a multi-player setting.
In other words, rather than consider the global structure of the relationship between the game parameters
and player behavior, we focus on the local structure of that relationship.

This analysis of the local structure allows us construct general theorems on
when there is a change to an information structure that will  \emph{reduce} information 
available to a player but \emph{increase} their expected utility. It also allows us to 
construct extended ``Pareto" versions of these theorems, specifying when 
there is a change to an information structure that will both reduce information 
available to \emph{all} players and increase \emph{all} of their expected utilities.

We illustrate these theoretical results with computer experiments involving
the noisy leader-follower game. We also discuss the general implications of these results for
well-known issues in the economics of information.

\subsection{Value of information}
\label{sec:val_of_info}
  
Intuitively, it might seem that a rational decision maker cannot be
hurt by additional information. After all, that is the standard
interpretation of Blackwell's famous result that adding noise to an
observation by sending it through an additional channel, called {\em
  garbling}, cannot improve expected utility of a Bayesian decision
maker in a game against Nature~\cite{blackwell1953equivalent}. However
games involving multiple players, and/or bounded rational behavior,
might violate this intuition.

To investigate the legitimacy of this intuition for general
noncooperative games, we first need to formalize what it means to have
``additional information". To begin, consider the simplest case, of a
single-player game. We can compare two scenarios: One where the player
can observe a relevant state of Nature, and another situation that is
identical, except that now she \emph{cannot} observe that state of
Nature. More generally, we can compare a scenario where the player
receives a noisy signal about the state of Nature to a scenario that
is identical except that the signal she receives is strictly noisier (in a certain sense)
than in the first scenario.  Indeed, in his seminal paper
\cite{blackwell1953equivalent}, Blackwell proved that the set of 
changes to an information channel that can never increase the expected utility of the player
are precisely those that are equivalent to sending the signal through an additional channel, and thereby introducing
extra noise.
So at least in a game against Nature, one can usefully define the ``value
of information'' as the difference in highest expected utility that can be
achieved in a low noise scenario (more information) compared to a high noise scenario (less 
information)~\cite{howard2005influence}, and 
prove important properties about this value of information.

In trying to extend this reasoning from a single player game to a
multi-player game two complications arise. First, in a multi-player game
there can be multiple equilibria, with different expected utilities
from one another. All of those equilibria
will change, in different ways, when noise is added to an information channel
connecting players in the game. Indeed, even the
number of equilibria may change when noise is added to a channel. This means
there is no well-defined way to compare equilibrium behavior in a ``before" scenario with equilibrium behavior
in an ``after" scenario in which noise
has been added; there is arbitrariness in which pair of equilibria, one from each scenario, we use for the comparison.
Note that there is no such ambiguity in a game against Nature. (In addition,
this ambiguity does not arise in the Cournot scenarios discussed below if we restrict
attention to subgame perfect equilibria.)

A second complication is that in a multi-player game \emph{all} of the players will react to
a change in an information channel, if not directly then indirectly,  via
the strategic nature of the game. This 
effect can even result in a negative value of information, in that it means
a player would prefer less (i.e., noisier) information. Indeed, such negative
value of information can arise even when both the ``before" and ``after" scenarios
have unique (subgame perfect) equilibria, so that there is no ambiguity in choosing which two
equilibria to compare.

To illustrate this, consider the Cournot duopoly where two competing manufacturers of a good simultaneously
choose a production level. Note that, as far as its equilibrium
structure is concerned, this scenario is equivalent to one player --- the ``leader" --- choosing
his production level first, but the other player, the ``follower", having no information 
about the leader's choice before making her choice. Assuming that both players can
produce the good for the same cost and that the demand function is linear, it
is well known that in that equilibrium both players get the same
profit. 

Now change the game by having the follower observe the
leader's move before she moves. So the only change is that the follower now has
more information before making her move. In this new game, the leader can choose
a higher production level compared to the production level of the simultaneous
move game --- the monopoly production level --- and the follower
has to react by choosing a lower production level. Thus, the follower is
actually hurt by this change to the game that results in her having more information. 

In this example, the leader changes
his move to account for the information that (he knows that) the follower will receive. Then, after receiving
this information, the follower cannot credibly
ignore it, i.e., cannot credibly behave as in the simultaneous move game equilibrium. So this equilibrium of the new game, where the
follower is hurt by the extra information, is subgame-perfect.
These and several other examples of negative value of information can
be found in the game theory literature (see section~\ref{sec:PreviousWork} for references). 

In this paper we introduce a broad framework that overcomes these two
complications which distinguish multi-player games from single-player
games. This framework is based on generalizing the concept of the
``marginal value of a good", to a decision-maker in a game against
Nature, so that it can apply to multi-player game scenarios. This
means that in our approach, the ``before" and ``after" scenarios
traditionally used to define value of information in games against
Nature are infinitesimally close to one another. More precisely, we
consider how much the expected utility of a player changes {\em per
  unit change in the amount of information} as one infinitesimally
changes the conditional distribution specifying the information
channel in a game. (This is illustrated in
Fig.~\ref{fig:marginal_value}.) Interestingly, such a local picture is
widely used in decision theory, i.e. marginal value, but has, to our
knowledge not been considered in game theory. Instead, as in the
example above, investigations of value of information have been based
exclusively on comparing global changes to the information structure,
e.g. no vs. full information.
 \noindent
 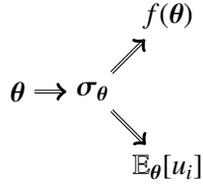
\begin{figure}[tbp]
   \begin{center}
     \begin{tikzpicture}
       \node (T) at (0,1) {$\vec{\theta}$};
       \node (S) at (1,1) {$\vec{\sigma}_{\vec{\theta}}$};
       \node (V) at (2,0) {$\E_{\vec{\theta}}[u_i]$};
       \node (f) at (2,2) {$f(\vec{\theta})$};
       \draw[->,double] (T) to (S);
       \draw[->,double] (S) to (V);
       \draw[->,double] (S) to (f);
       \end{tikzpicture}
   \end{center}
        \noindent \caption{Both the expected utility of player $i$ and the amount of information player $i$ receives
        depend, in part, on the strategy profile of all the players, $\vec{\sigma}$. Via the equilibrium concept, that profile in turn depends 
        on the specific conditional distributions $\vec{\theta}$
        in the information channel providing data to player $i$. So a change to $\vec{\theta}$
        results in a coupled change to the expected utility of player $i$, $\E_{\vec{\theta}}[u_i]$, and to the amount of information in their channel, $f({\vec{\theta}})$.
        The ``marginal value of information" to $i$ is how much $\E_{\vec{\theta}}(u_i)$ changes 
        per unit change in $f(\vec{\theta})$ if $\vec{\theta}$ is changed infinitesimally.}
 \label{fig:marginal_value}
 \end{figure}

 In the next subsection we provide a careful motivation for our
 ``marginal value of information" approach. As we discuss in the
 following subsection, this requires us to choose a cardinal measure
 of amount of information as well as a way to relate infinitesimal
 changes in utility to infinitesimal changes in information. Next we
 discuss the broad benefits of our approach, e.g., as a way to
 quantify marginal rates of substitution of different kinds of
 information arising in a game.  After this we relate our framework to
 previous work. We end this section by providing a roadmap to the rest
 of our paper.

\subsection{From consumers making a choice to multi-player games}
\label{sec:generalize_marg_util}

To motivate our framework, first consider the
simple case of a consumer whose preference function depends jointly on
the quantities of all the goods they get in a market. Given some
current bundle of goods, how should we quantify the value they assign to getting
more of good $j$? The usual economic answer is the
\emph{marginal utility of good $j$} to the consumer, i.e., the
derivative of their expected utility with respect to amount of good
$j$.

Rather than ask what the marginal value of good $j$ is to the
consumer, we might ask what their marginal value is for some linear
combination of $j$ and a different good $j'$. The natural answer is
their ``marginal value'', i.e. is the marginal utility of that precise
linear combination of goods.

More generally, rather than considering the marginal value to the
consumer of a linear combination of the goods, we might want to
consider the marginal value to them of some arbitrary, perhaps
non-linear function of quantities of each of the goods. What marginal
value would they assign to \emph{that}?

To ground our thinking, consider an example
where the function is the total weight of all the goods.
What would be the ``value of weight''? To answer this
question, write the bundle of goods the consumer possesses as
$\vec{\theta}$, i.e. $\theta^j$ is the amount of good $j$. Then write
the consumer's expected utility as $V(\vec{\theta})$, and the weight
of the bundle as the function $f(\vec{\theta})$. We want to relate how the
utility changes when the weight of the bundle is changed
infinitesimally. 

In general, there is no unique such relation since the weight $f$ is a
scalar-valued function and thus different changes $\delta
\vec{\theta}$ to the bundle could lead to the same change in the
weight of the bundle, even though they might result in different
changes to the expected utility. So to answer our question we need to
fix the direction of change $\delta \vec{\theta}$ to the bundle of
goods. Once we have done that, then to first-order the utility changes
by $\langle \grad V, \delta \vec{\theta}\rangle$ while the weight is
changed by $\langle \grad f, \delta
\vec{\theta}\rangle$.{\footnote{Here and throughout, $\langle \cdot,
    \cdot \rangle$ denotes a dot product, $\grad f$ is the gradient of
    $f$ and $\langle \grad f, \delta \vec{\theta}\rangle$ gives the
    length of the gradient vector when projected onto the direction
    $\delta\vec{\theta}$. }}

Accordingly we define the
``differential value of $f$ (in the direction $\delta \vec{\theta}$)'' as
\begin{equation}
  \frac{\langle \grad V, \delta \vec{\theta}\rangle}{\langle \grad f,
    \delta \vec{\theta}\rangle}
  \label{eq:proposed_measure}
\end{equation}
Thus, a change to the bundle has a high differential value of $f$, e.g. of weight, 
when that change (in the direction $\delta \vec{\theta}$) results in a large gain in utility for a small
change in $f$.

In some situations one might wish to speak of the ``value of $f$" without specifying a
specific direction of change to the vector-valued argument of $f$. Indeed, all
previous quantifications of ``value of information" that we know of, both in
game theory and in other fields (e.g., analysis of influence diagrams~\cite{howard2005influence,komi03}), quantifies
value of information without specifying a direction.

We can extend our approach to provide such a ``direction-free"
quantification of value of information, by choosing the direction
$\delta\vec{\theta} = \grad f$ which results in the largest change of
$f$, i.e. steepest ascent. Thus, we can quantify the ``differential
value of $f$'' as
\begin{eqnarray}
\frac {\langle \grad(V), \grad f \rangle}{|| \grad f ||^2} 
\label{eq:proposed_measure_2}
\end{eqnarray}
The measure in~\eqref{eq:proposed_measure_2} says that if a small
change in the value of $f$ leads to a big change in expected utility,
$f$ is more valuable than if the same change in expected utility
required a bigger change in the value of $f$. Furthermore, it
expresses the value of information in units of value (of an
infinitesimal change to $f$) \emph{per unit of $f$}. Thus, we follow
the conventional terminology where we would typically say that how
much the consumer values a change to good $j$ is given by the
associated change in utility \emph{divided by the amount of change in
  good} $j$.  (After all, that tells us change in utility \emph{per
  unit change in the good}.)

All of the reasoning above can be carried over
from the case of a single consumer to apply to multi-player scenarios. To see how, first
note that in the reasoning above, $\vec{\theta}$ is simply the parameter vector determining the utility of the consumer.
In other words, it is the parameter vector specifying the details of
a game being played by a decision maker in a game against Nature.
So it naturally generalizes to a multi-player game, as the parameter vector specifying the details of such
a game.

Next,  replace the consumer player in the reasoning above
by a particular player in the multi-player game. The key to the reasoning above is that 
specifying $\vec{\theta}$ specifies the expected utility of the consumer player.
In the case of the consumer, that map from parameter vector to expected utility
is direct. In a multi-player game, that direct map becomes an indirect map
specified in two stages: First by the equilibrium concept, taking $\vec{\theta}$
to the mixed strategy profile of all the players, and then from that profile to the expected utility of any
particular player. (Cf., Fig.~\ref{fig:marginal_value}.)

As mentioned above though, there is an extra complication in the multi-player case that is
absent in the case of the single consumer. Typically multi-player games have multiple equilibria for any $\vec{\theta}$, and 
therefore multiple values of $V(\vec{\theta})$. (In Fig.~\ref{fig:marginal_value}, the map from $\vec{\theta}$ to the
mixed strategy profile is multi-valued in games with multiple players.) However we need
to have the mapping from $\vec{\theta}$ to the expected utility of the
players be single-valued to use the reasoning above.  This means that we have to be careful when
calculating derivatives
to specify which precise branch of the set of equilibria we are
considering. Having done that, our generalization from the definition of marginal utility
for the case of a consumer choosing a bundle of goods to marginal utility for a player
in a multi-player game is complete.

\subsection{General comments on the marginal value approach}

There are several aspects of this general framework that are important to emphasize. First,
there is no reason to restrict attention to Nash equilibria (or some appropriate refinement). All
that we need is that $\vec{\theta}$ specifies (a set of) equilibrium expected utilities for
all the players. The equilibrium concept can be anything at all.

Second, note that $\vec{\theta}$, together with the
solution concept and choice of an equilibrium branch, specifies the mixed strategy profile of the players, as well
as all prior and conditional probabilities. So it specifies
the distributions governing the joint distribution over all random
variables in the game.  Accordingly, it specifies the values of all
cardinal functions of that joint distribution. So in particular, however
we wish to quantify ``amount of information", so long as it is a function
of that joint distribution, it is an indirect function of $\vec{\theta}$ (for a fixed
solution concept and associated choice of a solution branch).
This means we can  apply our analysis for any such quantification of the amount
of information as a function $f(\vec{\theta})$.

 We have to make several choices every time we use our approach.
 One is that we must choose what parameters of the game
 to vary. 
In addition, when analyzing value of information in a particular direction,
we have to specify that direction.  Taken together these choices fix what
economic question we are considering. Similar choices (e.g., of what
game parameters to allow to vary) arise, either implicitly or
explicitly, in any economic modeling.
Furthermore, we must decide 
what information measures we wish to analyze. (We address this issue in the next subsection.)

Finally, if we wish to analyze value of information 
we confront an additional, purely formal choice, unique to analyses of marginal values. 
This is the choice of what coordinate system to use to evaluate the dot products and gradients in Eq.~\eqref{eq:proposed_measure_2}.
The difficulty is that changing the coordinate system changes the
values of both dot products{\footnote{To give a simple example that the dot product can change
depending on the choice of coordinate system,
    consider the two Cartesian position vectors $(1, 0)$ and $(0, 1)$. Their
    dot product in Cartesian coordinates equals 0.  However if we
    translate those two vectors into polar coordinates we get $(1, 0)$
    and $(1, \pi /2)$.  The dot product of these two vectors is $1$,
    which differs from $0$, as claimed.}} and gradient vectors
in general.  So different choices of
coordinate system would give different marginal values of information. 
However since the choice of coordinate system is purely a modeling choice,
we do not want our conclusions to change if we change
how we parametrize the noise level in a communication channel, for
example. 

To address this problem we need to introduce a metric, as in
differential geometry.  In
different economic scenarios, different choices of metric may be
appropriate. In general, they can change our results.

Here, for reasons
of space, we concentrate on metric-independent results. Differential value of information
at a particular $\vec{\theta}$ depends on the choice of metric in general.{\footnote{Indeed, 
in some circumstances changing the metric can change
differential value of information from being positive to being negative.
This is just a specific instance of the general phenomena that the
change in the value of a function $V(\vec{\theta})$ for an infinitesimal
step in $\vec{\theta}$ along a gradient of a second function $f(\vec{\theta})$
may change from positive to negative, for certain
kinds of change to the metric. This is ultimately due to the fact that the metric specifies 
the shape of the ellipsoid specifying all points $\vec{\theta}'$
that lie an ``infinitesimal distance" from some current point $\vec{\theta}$, and
therefore specifies the direction of the gradient of $f(\vec{\theta})$.}}  However
the differential value of information in a given direction does not
depend on the metric (or coordinate system) as explained in section~\ref{sec:DiffValue}. 
Accordingly, in this paper we restrict attention to differential
value of information in a given direction.

\subsection{How to quantify information}
\label{sec:how_quantify_info}

To use the framework outlined in Sec.~\ref{sec:generalize_marg_util}, we must choose a function $f(\vec{\theta})$ that
measures the amount of information in a game with parameter $\vec{\theta}$.
Traditionally, a player's information is represented in game theory by a signal that
the player receives during the game\footnote{This includes information
  partitions, in the sense that the player is informed about which
  element of her information partition obtains.}. Thus information is
often thought of as an object or commodity. But this general approach does not integrate
the important fact that a signal is only informative to the
extent that it changes what the player believes
\emph{about some other aspect of the game}. It is the 
relationship between the value of the signal and that other aspect of the game
that determines the ``amount  of information" in the signal.

More precisely, let $y$, sampled from a distribution $p(y)$, be a payoff-relevant variable whose
state player $i$ would like to know before making her move, but which she cannot observe directly. 
Say that the value $y$  is used to generate a datum/signal $x$, that the player directly observes, via a conditional
distribution $p(x \mid y)$ which we typically refer to as an {\em information channel}. If for some reason
the player ignored $x$, then she would assign the \emph{a-priori} likelihood $p(y)$ to $y$,
even though in fact its \emph{a-posteriori} likelihood is $p(y \mid x) = \frac{p(y) p(x \mid y)}{\sum_{y'}p(y') p(x \mid y')}$. 
This change in the likelihoods she would assign to
$y$ is a measure of the information that $x$ provides about $y$. Arguably, this change of distribution
is the core property  of information that is of interest in economic scenarios. 

Averaging over observation $x$ and state $y$'s, and working in log space, this change in the 
likelihood she would assign to the actual $y$ if she ignored $x$ (and so used likelihoods $p(y)$ rather than $p(y \mid x)$) is

\begin{eqnarray}
\label{mutual-deriv}
\sum_{x,y} p(x) p(y \mid x) \ln\bigg[\frac{p(y \mid x)}{p(y)} \bigg] .
\end{eqnarray}

Thus, eq.~\eqref{mutual-deriv} gives the (average) increase in information that player $i$ has about $y$ due to observing $x$.
Note that this is true no matter how the variables $X$
and $Y$ arise in the strategic interaction. In particular, this interpretation of the quantity
in Eq.~\eqref{mutual-deriv} does not require that the value $x$ arise directly through 
a pre-specified distribution $p(x \mid y)$. $x$ and $y$ could instead be
variables concerning the strategies of the players at a particular equilibrium. 

In this sense, we have shown that the quantity in Eq.~\eqref{mutual-deriv} 
is a proper way to measure  the
information relating any pair of variables arising in a strategic scenario. None of the
usual axiomatic arguments motivating Shannon's information theory~\cite{mack03,coth91} were used in
doing this. 
Nevertheless, the quantity in Eq.~\eqref{mutual-deriv} is just the
mutual information between $X$ and $Y$, as defined by Shannon.
Mutual information, together with the related quantity of entropy,
forms the basis of information theory. It not only allows us to quantify
information, but has many applications in different areas ranging from
coding theory to machine learning to evolutionary biology.

In game theory, information is more commonly expressed in terms of an
information partition. An information partition $\{\set{A}_1, \ldots,
\set{A}_n\}$ (on some measure space $\Omega$) can be viewed as a
random variable with values $x \in \set{X} = \{1, \ldots, n\}$,
i.e. the signal $x$ reveals which element of the partition was hit.
Coarsening that partition can then be considered as a (deterministic)
information channel $p(y \mid x)$ from $x \in \set{X}$ to a value $y
\in \set{Y}$ in the coarser partition.  Now when we want to evaluate
how much information the agent obtains from the coarser partition $Y$
about some other random variable $N$, e.g. corresponding to a state of
Nature, we see that $N \rightarrow X \rightarrow Y$ is a Markov
chain\footnote{$X \to Y \to Z$ are said to form a Markov chain when
  $p(x,y,z) = p(x) p(y \mid x) p(z \mid y)$.}.  Thus, the {\em
  data-processing inequality}\footnote{The data processing inequality
  is an important theorem in information theory which shows that
  information can not be increased by processing it (via an
  information channel $p(z \mid y)$).
  \begin{theorem*}
    Let $X \to Y \to Z$ form a Markov chain. Then, $I(X ; Y) \geq I(X
    ; Z)$
 \end{theorem*}} applies and the mutual information between $N$ and
$Y$ cannot exceed the mutual information between $N$ and $X$. So by
using mutual information, we can not only state that the amount of
information is reduced when an information partition is coarsened, but
also quantify by how much. Furthermore, in this way, mutual
information provides a cardinal measure that is compatible with the
partial order based on coarsening information partitions.

When the distributions of $X$ or $Y$ are not fixed, as they might
correspond to moves of players in a game, it it more natural to
consider the information than can potentially transmitted by the
channel. In information theory this is measured by the {\bf channel capacity} $C$ of
an information channel $p(y \mid x)$ from $X$ to $Y$ as $C =
\max_{p(X)} I( X ; Y)$.  Unfortunately, in general we cannot solve the
maximization problem defining information capacity analytically. So
closed formulas for the channel capacity are only known for special
cases. This in turn means that partial derivatives of the channel
capacity with respect to the channel parameters are difficult to
calculate in general.  One special case where one can make that
calculation is the binary (asymmetric)
channel\cite{doi:10.1117/12.609436}.  For this reason, we will use
that channel in the examples considered in this paper that involve
marginal value of information capacity.\footnote{Another important
  class of information channels with known capacity are the so called
  symmetric channels \cite{coth91}. In this case, the noise is
  symmetric in the sense that it does not depend on a particular
  input, i.e. the channel is invariant under relabeling of the
  inputs. This class is rather common in practice and includes
  channels with continuous input, e.g. the Gaussian channel.}

While mutual information and channel capacity will be the typical
choices of $f$ in our computer experiments presented below, our
general theorems hold for arbitrary choices of $f$, even those that
bear no relation to concepts from Shannon's information theory.  Note
that even once we decide to use information theoretic quantities of a
signal to quantify information, we must still make the essentially
arbitrary choice of \emph{which} signal, to which player, concerning
which other variable, we are interested in. So for example, we might
be interested in the mutual information between some state of Nature
and the move of player 1.  Or the channel capacity between the
original move of player 1 and the last move of player 1.

\subsection{Benefits of the marginal value approach}

This approach of making infinitesimal changes to information channels
and examining the ramifications on expected utility is very
general and can be applied to any information channel in
the game. That means for example that we can add infinitesimal noise to
an information channel that models dependencies between different states of Nature and
examine the resultant change in the expected utility of a player. As another example,
we can change the information channel between two of the players in the game, and
analyze the implications for the expected utility of a third player in the game.

In fact, the core idea of this approach extends beyond making infinitesimal
changes to the noise in a channel. At root, what we are doing is making an infinitesimal change
to the parameter vector that specifies the noncooperative game. This 
differential approach can be applied to other kinds of
infinitesimal changes besides those involving noise vectors in
communication channels. For example, it can be applied to a change to
the utility function of a player in the game (see example \ref{ex:neg_val_ut}). As another example, the changes
can be applied to the rationality exponent of a player under a logit quantal response
equilibrium~\cite{mcpa98}. 
This flexibility allows us to extend Blackwell's idea of ``value of
information'' far beyond the scenarios he had in mind, to (differential) value of \emph{any defining
characteristics of a game.}
This in turn allows us to calculate marginal rates of substitution of any component of a game's
parameter vector with any other component, e.g., the marginal rate of substitution for player $i$ of (changes to)
a specific information channel and of (changes to) a tax applied to player $j$.

More generally still, there is nothing in our framework that requires us to consider marginal values \emph{to
a player in the game}. So for example, we can apply our analysis to calculate marginal social welfare of (changes to)
information channels, etc. Carrying this further, we can use our framework to calculate marginal rates of substitution in
noncooperative games to an external regulator concerned with social welfare who is able to change some parameters
in the game.

In this context, the need to specify a particular branch of the game is a benefit
of the approach, not a shortcoming. To see why, consider how a (toy model of a regulator) concerned with social welfare
would set some game parameters, according to conventional economics analysis. 
The game and associated set of parameter
vectors is considered \emph{ab initio}, and
an attempt is made to find the global optimal value of the parameter vector. However
whenever the game has multiple equilibrium branches, in general what parameter vector
is optimal will depend on which branch one considers --- and there is
no good generally applicable way of predicting which branch
will be appropriate, since that amounts to choosing a universal refinement.

However our framework provides a different way for the regulator to
control the parameter vector. The idea is to start with the
\emph{actual} branch that gives an actual, current player profile for
a currently implemented parameter vector $\vec{\theta}$.  We then tell
the regulator what direction to incrementally change that parameter
vector \emph{given that the players are on that branch}. No attempt is
made to find an \emph{ab initio} global optimum. So this approach
avoids the problem of predicting what branch will arise --- we use the
one that \emph{is actually occurring}. Furthermore, the parameters
can then be changed along a smooth path leading the players from the
current to the desired equilibrium (see \cite{woha12} for an example
of this idea).

\subsection{Previous work}
\label{sec:PreviousWork}

In his famous theorem, Blackwell formulated the imperfect information of the
decision maker concerning the state of Nature as an information
channel from the move of Nature to the observation of the decision
maker, i.e., as conditional probability distribution, leading from the
move of Nature to the observation of the decision maker. This is a
very convenient way to model such noise, from a calculational
standpoint. As a result, it is the norm for how to formulate imperfect
information in Shannon information theory~\cite{coth91,mack03}, which
analyses many kinds of information, all formulated as real-valued
function of probability distributions.  Indeed, use of conditional
distributions to model imperfect information is the norm in all of
engineering and the physical sciences, e.g., computer science, signal
processing, stochastic control, machine learning, physics, stochastic
process theory, etc.

There were some early attempts to use Shannon information theory in
economics to address the question of the value of information. Except for
special cases such as multiplicative payoffs (Kelly gambling
\cite{Kelly1956}) and logarithmic utilities \cite{Arrow1971}, where the
expected utility will be proportional to the Shannon entropy, the use
of Shannon information was considered to provide no additional insights. Indeed,
\citet{Radner1984} rejected the use of any single
valued function to measure information because it provides a total
order on information and therefore allows for a negative value of
information even in the decision case considered by Blackwell.

In multi-player game theory, i.e. multi-agent decision situations,
the role of information is even more involved. Here,
many researchers have constructed special games, showing that the
players might prefer more or less information depending on the
particular structure of the game (see \cite{lepp77} for an early
example). This work showed that Blackwell's result cannot directly be
generalized to situations of strategic interactions.

Correspondingly, the most common formulation of imperfect information
in game theory does not use information channels let alone Shannon
information. Instead, states of Nature are lumped using information
partitions specifying which states are indistinguishable to an
agent. In this approach, more (less) information is usually modeled as
refining (coarsening) an agent's information partition.  In
particular, noisy observations are formulated using such partitions in
conjunction with a (common) prior distribution on the states of
Nature. Even though, this is formally equivalent to conditional
distributions, it leads to a fundamentally different way of thinking
about information. The formulation of information in terms of
information partitions provides a natural partial order based on
refinining partitions. Thus, in contrast to Shannon information
theory, which quantifies the amount of information, it cannot compare
the information whenever the corresponding partitions are not related
via refiniments. In addition, the avoidance of conditional distributions
makes many calculations more difficult.

Recently, some work in game theory has made a distinction between the
``basic game''and the ``information structure''\footnote{According to
  \citet{Gossner2000} this terminology goes back to Aumann.}: The
\emph{basic game} captures the available actions, the payoffs and
the probability distribution over the states of Nature, while the
information structure specifies what the players believe about the
game, the state of Nature and each other (see for instance
\cite{Lehrer2013,Bergemann2013a}).  More formally this is expressed in
games of incomplete information having each player observing a signal,
drawn from a conditional probability distribution, about the state of
Nature. In principle these signals are correlated. The effects of
changes in the information structure were studied by considering
certain types of garbling nodes as by Blackwell. While this goes
beyond refinements of information parttions, it still only provides a
partial order of information channels.

\citet{Lehrer2013} showed that if two information
structures are equivalent with respect to a specific garbling the game
will have the same equilibrium outcomes. Thus, they characterized the
class of changes to the information channels that leave the players
indifferent with respect to a particular solution concept. Similarly,
\citet{Bergemann2013a} introduced a Blackwell-like
order on information structures called ``individual sufficiency" that
provides a notion of more and less informative structures, in the sense
that more information always shrinks the set of Bayes correlated
equilibria. A similar analysis relating the set of equilibria between
different information structures has been obtained by \citet{Gossner2000} and is in line with his work
\cite{DBLP:journals/geb/Gossner10} relating more knowledge of the
players to an increase of their abilities, i.e. the set of possible
actions available to them. As formulated in this work, more information can be seen to
increase the number of constraints on a possible solution for the
game.

Overall, the goal of these attempts has been to characterize changes
to information structures which imply certain properties of the
solution set, independent of the particular basic game. This is
clearly inspired by Blackwell's result which holds for all possible
decision problems. So in particular, these analyses aim for results
that are independent of the details of the utility function(s). Moreover, the
analyses are concerned with results that hold simultaneously for all
solution points (branches) of a game. Given these constraints on
the kinds of results one is interested in, as observed by Radner and Stiglitz, Shannon
information (or any other quantification of information) is not of
much help.

In contrast, we are concerned with analyses of the role of information in
strategic scenarios that concern a particular game
with its particular utility functions. Indeed, our analyses focus on
a single solution point at a time, since the role of information for the
exact same game game will differ depending on which solution branch
one is on. Arguably, in many scenarios regulators and analysts of a strategic
scenario are specifically interested in the \emph{actual} game being played,
and the \emph{actual} solution point describing the behavior of its players. As
such, our particularized analyses can be more relevant than broadly applicable
analyses, which ignore such details.

While not being much help in the broadly applicable analyses of
\citet{Bergemann2013a,Gossner2000,DBLP:journals/geb/Gossner10}, etc., we argue below that
Shannon information \emph{is} useful if one
wants to analyze the role of information in a particular game with its
specific utility functions. In this case, the idea of marginal utility
of a good to a decision-maker in a particular game against Nature
can be naturally extended ``marginal utility"  of information to a player
in a particular multi-player game on a particular solution branch of that game. Thus, one is naturally lead to a quantitative notion of
information and the differential value of information as elaborated
above.

\subsection{Roadmap}

In Sec.~\ref{sec:MAID}, we review Multi-Agent Influence Diagrams (MAIDs)
and explain why they are especially suited to study information in
games. Next, we introduce quantal response equilibria of MAIDs and show how to
calculate partial derivatives of the associated strategy profile with respect to components of the associated game parameter vector. 

Based on these definitions, in Sec.~\ref{sec:information_economics} we
define the differential value of information and in
Sec.~\ref{sec:properties} we prove that generically, in all games
there is such a direction in which information is decreased. Here,
generic refers to the parametrization of the game, i.e. we do not
assume any \emph{a priori} constraints on how the channel's
conditional distribution can be changed. In this sense, we prove that
in all games, i.e. for any type of utility, e.g. zero-sum, and any
structure, simultaneous or sequential, there is (a way to
infinitesimally change the channel that has) negative value of
information.
Similarly, we provide necessary and sufficient conditions for
a game to have negative value of information simultaneously for all players. (This 
condition can be viewed as a sort of``Pareto negative value of information''.)

Next, in Sec.~\ref{sec:Examples} we illustrate our proposed
definitions and results in a simple decision situation as well as an
abstracted version of the duopoly scenario that was discussed above,
in which the second-moving player observes the first-moving player
through a noisy channel.  In particular, we show that as one varies
the noise in that channel, the differential value of information is indeed
sometimes negative for the second-moving player, for certain starting
conditional noise distributions in the channel (and at a particular
equilibrium). However for other starting distributions in that channel
(at the same equilibrium), the differential value of information is positive
for that player. In fact, all four pairs of \{positive / negative\}
marginal value of information for the \{first / second\} -- moving
player can occur.

After this we present a section giving more examples. We end with a discussion
of future work and conclusions.

A summary of the notation we use is provided in Table.~\ref{tab:notation}. 

\begin{table}
  \begin{tabular}{rcl}
    \multicolumn{3}{l}{\underline{Information theory}} \\
    $\set{X}, \set{Y}$ & & Sets \\
    $x, y$ & & Elements of sets, i.e. $x \in \set{X}$ \\
    $X, Y$ & & Random variables with outcomes in $\set{X}, \set{Y}$ \\
    $\Delta_{\set{X}}$ & & Probability simplex over $\set{X}$. \\
    $I(X;Y)$ & & Mutual information between $X$ and $Y$ \\
    \multicolumn{3}{l}{\underline{Differential geometry}} \\
    $\vec{v}, \vec{\theta}$ & & Vectors \\
    $v^i$ & & $i$-th entry of contra-variant vector \\
    $v_i$ & & $i$-th entry of co-variant vector \\
    $g_{ij}$ & & Metric tensor. Its inverse is denoted by $g^{ij}$. \\
    $\frac{\partial}{\partial \theta^i}$ & & Partial derivative wrt/ $\theta^i$ \\
    $\grad(f)^i = g^{ij}\frac{\partial}{\partial \theta^i}$ & & Gradient of $f$, i.e. contra-variant direction of steepest ascent \\ 
    $\nabla \nabla f$ & & Hessian of $f$ \\
    $\langle \vec{v},\vec{w} \rangle_g$ & & Scalar product of $\vec{v}, \vec{w}$ wrt/ metric $g$ \\
    $|\vec{v}|_g$ & & Norm of vector $\vec{v}$ wrt/ metric $g$ \\
    \multicolumn{3}{l}{\underline{Multi-agent influence diagrams}} \\
    $G = (\set{V}, \set{E})$ & & Directed acyclic graph with vertices $\set{V}$ and edges $\set{E} \subset \set{V} \times \set{V}$ \\
    $\set{X}_v$ & & State space of node $v \in \set{V}$ \\
    $\set{N}$ & & Set of Nature or change nodes, i.e. $\set{N} \subset \set{V}$ \\
    $\set{D}_i$ & & Set of decision nodes of player $i$ \\
    $pa(v) = \{ u : (u,v) \in \set{E} \}$ & & Parents of node $v$ \\
    $p(x_v \mid x_{pa(v)})$ & & Conditional distribution at Nature node $v \in \set{N}$ \\
    $\sigma_i(a_v \mid x_{pa(v)})$ & & Strategy of player $i$ at decision node $v \in \set{D}_i$ \\
    $u_i$ & & Utility function of player $i$ \\
    $\E( u_i \mid a_i )$ & & Conditional expected utility of player $i$ \\
    $V_i = \E(u_i)$ & & Value, i.e. expected utility, of player $i$ \\
    \multicolumn{3}{l}{\underline{Differential value of information}} \\
    $\set{V}_{\delta \vec{\theta}}$ & & Differential value of direction $\delta \vec{\theta}$ \\
    $\set{V}_{f, \delta \vec{\theta}}$ & & Differential value of $f$ in direction $\delta \vec{\theta}$ \\
    $\set{V}_{f}$ & & Differential value of $f$ \\
    $Con(\{\vec{v}_i\})$ & & Conic hull of nonzero vectors $\{\vec{v}_i\}$ \\
    $Con(\{\vec{v}_i\})^\bot$ & & Dual to the conic hull $Con(\{\vec{v}_i\})$
  \end{tabular}
  \caption{\label{tab:notation} Summary of notation used throughout the paper.}
\end{table}

\section{Multi-agent influence diagrams}
\label{sec:MAID}

Bayes nets~\cite{kofr09} provide a very concise, powerful way to model
scenarios where there are multiple interacting Nature players (either
automata or inanimate natural phenomena), but no human players. They
do this by representing the information structure of the scenario in
terms of a Directed Acyclic Graph (DAG) with conditional probability
distributions at the nodes of the graph. In particular, the use of
conditional distributions rather than information partitions
greatly facilitates the analysis and associated computation of the
role of information in such systems. As a result they have become very
wide-spread in machine learning and information theory in particular,
and in computer science and the physical sciences more generally.

Influence Diagrams (IDs~\cite{howard2005influence}) were introduced to
extend Bayes nets to model scenarios where there is a (single) human
player interacting with Nature players.  There has been much analysis
of how to exploit the graphical structure of the ID to speed up
computation of the optimal behavior assuming full rationality, which
is quite useful for computer experiments.

More
recently, Multi-Agent Influence Diagrams (MAIDs~\cite{komi03}) and
their variants like semi-net-form games~\cite{lewo12b,babe14,lewo12a} and Interactive POMDP's~\cite{doshi2009graphical} have
extended IDs to model games involving arbitrary numbers of players. As
such, the work on MAIDs can be viewed as an attempt to create a new
game theory representation of multi-stage games based on Bayes nets,
in addition to strategic form and extensive form representations.

Compared to these older representations, typically MAIDs more clearly express the
interaction structure of what information is available to each
player in each possible state.\footnote{In a MAID a player has information at a
decision node $A$ about some state of Nature $X$ if there is a
directed edge from $X$ to $A$.} They also very often require far less notation
than those other representations to fully specify a given game. Thus, we consider them as a natural
starting point when studying the role of information in games.

A MAID is defined as follows:
\begin{definition}
An $n$-player \textbf{MAID} is defined as a tuple $(G, \{\set{X}_v\},
\{p(x_v \mid x_{pa(v)})\}, \{u_i\})$ of the following elements:
\begin{itemize}
\item A directed acyclic graph $G = (\set{V}, \set{E})$ where $\set{V}
  = \set{D} \cup \set{N}$ is partitioned
  into
  \begin{itemize}
  \item a set of \emph{Nature} or \emph{chance} nodes $\set{N}$ and
  \item a set of decision nodes $\set{D}$ which is further
    partitioned into $n$ sets of \emph{decision} nodes $\set{D}_i$, one for
    each player $i = 1,\ldots,n$,
  \end{itemize}
\item a set  $\set{X}_v$ of states for each $v \in \set{V}$,
\item a conditional probability distribution $p(x_v \mid x_{pa(v)})$
  for each Nature node $v \in \set{N}$, where $pa(v) = \{ u : (u,v)
  \in \set{E}\}$ denotes the parents of $v$ and $x_{pa(v)}$ is their
  joint state.
\item a family of utility functions $\{u_i: \cart_{v \in \set{V}}
  \set{X}_v \to \R \}_{i = 1,\ldots,n}$.
\end{itemize}
\end{definition}
\noindent In particular, as mentioned above, a one-person MAID is an \textbf{influence diagram} (ID~\cite{howard2005influence}). 

In the following, the states $x_v \in \set{X}_v$ of a decision node $v \in
\set{D}$ will usually be called {\em actions} or \emph{moves}, and sometimes will be denoted by $a_v \in \set{X}_v$.
We adopt the convention that ``$p(x_v \mid x_{pa(v)})$'' means $p(x_v)$ if $v$ is
a root node, so that $pa(v)$ is empty. 
We  write elements of $\set{X}$ as $x$. We define $\set{X}_{\set{A}}
\equiv \prod_{v \in \set{A}} \set{X}_{v}$ for any $\set{A} \subseteq
\set{V}$, with elements of $\set{X}_{\set{A}}$ written as
$x_{\set{A}}$.  So in particular, $\set{X}_\set{D} \equiv \prod_{v \in
  \set{D}} \set{X}_v$, and $\set{X}_\set{N} \equiv \prod_{v \in
  \set{N}} \set{X}_v$, and we write elements of these
sets as $x_\set{D}$ (or $a_\set{D}$) and $x_{\set{N}}$, respectively.

We will sometimes write an $n$-player MAID as $(G, \set{X}, p,
\{u_i\})$, with the decompositions of those variables and associations
among them implicit.  (So for example the decomposition of $\set{G}$
in terms of $\set{E}$ and a set of nodes $[\cup_{i = 1,\ldots,n}
  \set{D}_i] \cup \set{N}$ will sometimes be implicit.)

A \emph{solution concept} is a map from any MAID $(G,
\set{X}, p, \{u_i\})$ to a set of conditional distributions
$\{\sigma_i(x_v \mid x_{pa(v)}) : v \in \set{D}_i, i =
1,\ldots,n\}$. We refer to the set of distributions $\{\sigma_i(x_v
\mid x_{pa(v)}) : v \in \set{D}_i\}$ for any particular player $i$ as
that player's \emph{strategy}. We refer to the full set
$\{\sigma_i(x_v \mid x_{pa(v)}) : v \in \set{D}_i, i = 1,\ldots,n\}$
as the \emph{strategy profile}.  We sometimes write $\sigma_v$ for a
$v \in \set{D}_i$ to refer to one distribution in a player's strategy
and use $\sigma$ to refer to a strategy profile.

The intuition is that each player can set the conditional distribution
at each of their decision nodes, but is not able to introduce
arbitrary dependencies between actions at different decision nodes. In
the terminology of game theory, this is called the agent
representation. The rule for how the set of all players jointly set
the strategy profile is the solution concept.
  
In addition, we allow the solution concept to depend on
parameters. Typically there will be one set of parameters associated
with each player. When that is the case we sometimes write the
strategy of each player $i$ that is produced by the solution concept
as $\sigma_i(a_v \mid x_{pa(v)} ; \vec{\beta})$ where $\vec{\beta}$ is
the set of parameters that specify how $\sigma_i$ was determined via
the solution concept.

The combination of a MAID $(G, \set{X}, p, \{u_i\})$ and a solution
concept specifies the conditional distributions at all the nodes of
the DAG $G$. Accordingly it specifies a joint probability distribution
\begin{eqnarray}
p(x_{\set{V}}) &=& \prod_{v \in \set{N}} p(x_v \mid x_{pa(v)}) \prod_{i = 1,\ldots,n}
  \prod_{v \in \set{D}_i} \sigma_i(a_v \mid x_{pa(v)}) \\
  & = & \prod_{v \in \set{V}} p(x_v \mid x_{pa(v)})
\end{eqnarray}
where we abuse notation and denote $\sigma_i(a_v \mid
x_{pa(v)})$ by $p(x_v \mid x_{pa(v)})$ whenever $v \in
\set{D}_i$.

In the usual way, once we have such a joint distribution over all
variables, we have fully defined the joint distribution over $\set{X}$
and therefore defined conditional probabilities of the states of one
subset of the nodes in the MAID, $\set{A}$, given the states of
another subset of the nodes, $\set{B}$:
\begin{eqnarray}
p(x_{\set{A}} \mid x_\set{B}) &=& \frac{p(x_{\set{A}}, x_{\set{B}})} {p(x_\set{B})} \nonumber \\
 &=& \frac{ \sum_{x_{\set{V} \setminus (\set{A} \cup \set{B})}} \; p(x_{\set{A} \cup \set{B}}, x_{\set{V} \setminus (\set{A} \cup \set{B}) }) }
 { \sum_{x_{\set{V} \setminus \set{B}}} \; p(x_{\set{B}}, x_{\set{V} \setminus  \set{B}) }) }
\end{eqnarray}
Similarly the combination of a MAID and a solution concept fully
defines the conditional value of a scalar-valued function of all
variables in the MAID, given the values of some other variables in the
MAID. In particular, the conditional expected utilities are given by
\begin{equation}
\E(u_i \mid x_{\set{A}} ) = \sum_{x_{\set{V} \setminus \set{A}}} \; p(x_{\set{V} \setminus \set{A}}  \mid x_{ \set{A}} ) u_i(x_{\set{V} \setminus \set{A}}, x_{ \set{A}} )
\label{eq:CondExp}
\end{equation}

We will sometimes use the term ``information structure'' to refer to
the graph of a MAID and the conditional distributions at its Nature
nodes. (Note that this is a slightly different use of the term from
that used in extensive form games.)  In order to study the effect of
changes to the information structure of a MAID, we will assume that
the probability distributions at the Nature nodes are parametrized by
a set of parameters $\vec{\theta}$, i.e., $p_v(x_v \mid x_{pa(v)};
\vec{\theta})$. We are interested in how infinitesimal changes to
$\vec{\theta}$ (and other parameters of the MAID like $\vec{\beta}$)
affect $p(x_{\set{V}})$, expected utilities, mutual information among nodes in
the MAID, etc.

\subsection{Quantal response equilibria of MAIDs}

A solution concept for a game specifies how the actions of the players
are chosen. In our framework, it is not crucial which solution concept
is used (so long as the strategy profile of the
players at any $\vec{\theta}$ is differentiable in the interior of $\Theta$). For convenience, we choose the (logit) quantal
response equilibrium (QRE) \cite{mcpa98}, a popular model for bounded
rationality.{\footnote{In addition, the QRE can be derived from information-theoretic principles~\cite{woha12},
although we do not exploit that property of QREs here.}}  Under a QRE, each player $i$
does not necessarily make the best possible move, but instead chooses
his actions at the decision node $v \in \set{D}_i$ from a Boltzmann
distribution over his move-conditional expected utilities: \be
\sigma_i(a_v \mid x_{pa(v)}) = \frac{1}{Z_i(x_{pa(v)})} e^{\beta_i
  \E(u_i | a_v, x_{pa(v)})}
\label{eq:QRE}
\ee for all $a_v \in \set{X}_v$ and $x_{pa(v)} \in \cart_{u \in pa(v)}
\set{X}_u$.  In this expression $Z_i(x_{pa(v)}) = \sum_{a \in
  \set{X}_{pa(v)}} e^{\beta_i \E(u_i | a, x_{pa(v)})}$ is a
normalization constant, $\E(u_i | a_v, x_{pa(v)})$ denotes the
conditional expected utility as defined in \equ{CondExp} and $\beta_i$
is a parameter specifying the ``rationality'' of player
$i$.

This interpretation is based on the observation that a player with
$\beta = 0$ will choose her actions uniformly at random, whereas
$\beta \to \infty$ will choose the action(s) with highest expected
utility, i.e., corresponds to the rational action choice. Thus, it
includes the Nash equilibrium where each player maximizes expected
utility as a boundary case. 
  
As shorthand, we denote the (unconditional) expected utility of player
$i$ at some equilibrium $\{\sigma_i\}_{i = 1,\ldots,n}$,
$\E{\{\sigma_i\}}_{i = 1,\ldots,n}(u_i)$, by $V_i$.
  
By using implicit differentiation, it is straight forward to compute
partial derivatives governed by the QRE solution concept (see
appendix~\ref{app:partial_derivative} for details). In particular, we
can compute $\frac{\partial}{\partial \theta^i}V_i$ for the QRE. This
is required for the definition of differential value of information
central to the analysis below.

\section{Differential value} 
\label{sec:information_economics}

\label{sec:DiffValue}

Say that we fix all distributions at Nature nodes in a MAID
except for some particular Nature-specified information channel $p(x_v \mid x_{pa(v)})$,
and are interested in the differential value of mutual information through that channel.
In general, the expected utility of a player $i$ in this MAID is not a single-valued function of 
the mutual information in that channel $I(X_v ; X_{pa(v)})$. There are two reasons for
this. First, the same value of $I(X_v ;
X_{pa(v)})$ can occur for different conditional distributions $p(x_v
\mid x_{pa(v)})$, and therefore that value of $I(X_v ; X_{pa(v)})$ can
correspond to multiple values of expected utility in general. Second, as discussed above, even
if we fix the distribution $p(x_v \mid x_{pa(v)})$,
there might be several equilibria (strategy profiles) all of which solve
the QRE equations but correspond to different distributions at the
decision nodes of the MAID. 

Evidently then, if $v$ is a chance node in a MAID and
$i$ a player in that MAID, there is no unambiguously defined ``differential value to $i$ of the mutual
information'' in the channel from $pa(v)$ to $v$. We can only talk about differential value of
mutual information \emph{at a particular joint distribution of the MAID},
a distribution that both specifies a particular equilibrium of player strategies on one particular equilibrium branch, and that
specifies one particular channel distribution $p(x_v \mid x_{pa(v)})$.
Once we make such a specification,
we can analyze several aspects of the associated value of mutual information.

Now assume that the channel $p(x_v | x_{pa(v)}; \vec{\theta})$ is
parameterized by $d$ parameters $\vec{\theta} = (\theta^1, \ldots,
\theta^d)$.\footnote{In the following, we make use of several
  conventions used in differential geometry (for a thorough introduction
  to differential geometry we refer the reader to a standard monograph
  \cite{Jost11}):
  \begin{itemize}
  \item {\em Covariant and contravariant vectors}: Vectors are called
    co- or contravariant depending on how their coordinate
    representations change under coordinate transformations. The
    convention is to use upper indices for contravariant and lower
    indices for covariant vectors.

    As an example consider a function $f(\vec{\theta})$ of the game
    parameters $\vec{\theta} = (\theta^1, \ldots, \theta^d)$. Then,
    $\vec{\theta}$ is contravariant, while the partial derivatives
    $(\frac{\partial f}{\partial \theta^1}, \ldots, \frac{\partial
      f}{\partial \theta^d})$ represents a covariant vector.
  \item {\em Einstein summation convention}: \bel{B1} v^i
    w_i:=\sum_{i=1}^d v^i w_i \qe The content of this convention is
    that a summation sign is omitted when the same index occurs twice
    in a product, once as an upper and once as a lower index.
  \item {\em Metric tensor}: The function of the metric tensor is to
    provide a (Euclidean) scalar product of tangent vectors,
    i.e. $\langle \vec{v},\vec{w} \rangle_g = g_{ij} v^i w^j$. 
    
    In particular, the gradient of a function $f(\vec{\theta})$ has
    contravariant coordinates $(\grad f)^i = g^{ij} \frac{\partial
      f}{\partial \theta^j}$ where $g^{ij}$ denotes the inverse of the
    metric tensor $g_{ij}$.
  \end{itemize}} Then, as motivated in
section~\ref{sec:generalize_marg_util} we want to study how
infinitesimal changes to the parameters effect the value of the
player. Consider a change in the direction $\delta \vec{\theta}$.  To
first order, this changes the value by $\langle \grad V, \delta
\vec{\theta} \rangle$ and correspondingly we can quantify the
differential value of this direction as the change in value per unit
length of $\delta \vec{\theta}$:
\begin{definition}
  Let $\delta \vec{\theta} \in \R^d$ be a contravariant vector. The
  \textbf{(differential) value of direction $\delta \vec{\theta}$} at
  $\vec{\theta}$ is defined as
\begin{eqnarray}
\set{V}_{\delta \vec{\theta}}(\vec{\theta}) &\equiv& \frac{ \langle \grad(V), \delta \vec{\theta}\rangle  }{ | \delta \vec{\theta} | } \nonumber
  \end{eqnarray}
\label{def:val_direction}
\end{definition}
\noindent This is the length of the projection of $\grad(V)$
in the unit direction $\delta \bftheta$. Intuitively, the direction $\delta
\bftheta$ is valuable to the player to the extent that $V$ increases in this
direction. This is what the value of direction $\delta \bftheta$ quantifies. (Note that
when $V$ decreases in this direction, the value is
negative.)

Using the definitions of gradients and scalar products we can expand
\begin{equation}
\set{V}_{\delta \vec{\theta}}(\vec{\theta}) = \frac{ \langle \grad(V), \delta \vec{\theta}\rangle  }{ | \delta \vec{\theta} | }
= \frac{  \frac{\partial}{\partial \theta^k} V 
  g(\vec{\theta})^{ki} g(\vec{\theta})_{il}
 \delta \theta^l  } {\sqrt{\delta \theta^k  g(\vec{\theta})_{kl}  \delta \theta^l}}
= \frac{  \frac{\partial}{\partial \theta^k} V 
  \delta \theta^k  } {\sqrt{\delta \theta^k  g(\vec{\theta})_{kl}  \delta \theta^l}}
\label{eq:val_direction}
\end{equation}
The absence of the metric in the numerator in
Eq.~\eqref{eq:val_direction} reflects the fact that the vector of
partial derivatives $\frac{\partial}{\partial \theta^k} V$ is a covariant vector, whereas $\delta \vec{\theta}$ is
a contravariant vector.

As discussed above and elaborated below, one important class of directions $\delta \vec{\theta}$ at a given game vector $\vec{\theta}$ are
gradients of functions $f(\vec{\theta})$ evaluated at $\vec{\theta}$, e.g., the direction $\grad I(X ; S)$.
However even when the direction $\delta \vec{\theta}$ we are considering is not
parallel to the gradient of an information-theoretic function $f(\vec{\theta})$ like mutual information, capacity or player
rationality, we will often be
concerned with quantifying the ``value" of such a $f$ in that direction $\delta \vec{\theta}$.  We can do this
with the following definition, related to the definition of differential value of a direction.
\begin{definition}
    Let $\delta \vec{\theta} \in \R^d$ be a contravariant vector. The
  \textbf{differential value of $f$ (in direction $\delta \vec{\theta})$} at
  $\vec{\theta}$ is defined as:
\begin{eqnarray}
  \set{V}_{f, \delta \vec{\theta}} & \equiv & \frac{\frac{ \langle \grad(V),
      \delta \vec{\theta}\rangle } { |\delta \vec{\theta}| }}{\frac { \langle \grad(f),
      \delta \vec{\theta}\rangle } { |\delta \vec{\theta}| }}
  = \frac{ \langle \grad(V),
    \delta \vec{\theta}\rangle }{ \langle \grad(f),
    \delta \vec{\theta}\rangle } 
  = \frac{ \frac{\partial}{\partial \theta^i} V \delta \theta^i}{ \frac{\partial}{\partial \theta^i} f \delta \theta^i}\nonumber
\end{eqnarray}
\label{def:val_f_in_direction}
\end{definition}
\noindent This quantity considers the relation between how $V$ and $f$ change when moving in the
direction $\delta \bftheta$.  If the sign of the differential value of
$f$ in direction $\delta \vec{\theta}$ at $\vec{\theta}$ is positive, then an
infinitesimal step in in direction $\delta \vec{\theta}$ at $\vec{\theta}$ will
either increase both $V$ and $f$ or decrease both of them. If instead
the sign is negative, then such a step will have opposite effects on
$V$ and $f$. The size of the differential value of $f$ in direction
$\delta \vec{\theta}$ at $\vec{\theta}$ gives the rate of change in $V$ per unit
of $f$, for movement in that direction. 
Note that $\set{V}_{f, \delta \vec{\theta}}$ is independent of the metric because both numerator and denominator are. 

Finally, we define the differential value of information of $f$
without explicitly specifying a direction. To do this we note that the
function $f(\theta)$ itself gives rise to a preferred direction,
namely the direction of steepest ascent of $f(\theta)$, i.e., the
direction of $\grad f$.  Thus, using $\delta \vec{\theta} = \grad f$
in definition~\ref{def:val_f_in_direction} we obtain:
\begin{definition}
  The \textbf{differential value of $f$} at $\vec{\theta}$ is defined
  as:
  \begin{equation}
    \set{V}_f \equiv \frac{\langle \grad V, \grad f\rangle}{\langle \grad V, \grad f\rangle}
    = \frac{\langle \grad V, \grad f\rangle}{|| \grad f ||^2}
  \end{equation}
\end{definition}
While this quantity does have units of value per unit of $f$, as
required by the interpretation as value of information, it is not
independent of the metric. Here, the numerator $\langle \grad V, \grad
f\rangle = \frac{\partial}{\partial \theta^i} V g^{ij}
\frac{\partial}{\partial \theta^j} f$ as well as the denominator $||
\grad f ||^2 = \frac{\partial}{\partial \theta^i} f g^{ij}
\frac{\partial}{\partial \theta^j} f$ are metric dependent. For this
reason, in the following we will focus on the metric independent
differential value of information in a direction.

\section{Properties of differential value}
\label{sec:properties}

We now present some general results concerning value of a function 
$f : \Theta \rightarrow \R$, in particular conditions for negative values. 
Throughout this section, we  assume  that both $f$ and $V$ are twice continuously differentiable. 
In addition, note that when we randomly and independently choose (the directions of) $n\le d$ vectors in $\R^d$, they are linearly independent with probability 1. That means, generically $n\le d$ nonzero vectors span an $n$-dimensional linear subspace. In the sequel, we shall often implicitly assume that we are in such a generic situation and refrain from discussing nongeneric situations, that is, situations with additional linear dependencies among the vectors involved. 

Note that generic here refers to the parametrization $\vec{\theta}$ in
that changes are not restricted to a lower-dimensional
subspace. Unless noted otherwise, we do {\em not} make any assumptions
about the utility function and our results are valid even in the
presence of non-generic symmetries, e.g. zero-sum, of the utility.

\subsection{Preliminary definitions}

To begin we introduce some particular convex cones (see
appendix~\ref{app:ConicHull} for the relevant definitions) that we
will use in our analysis of differential value of $f$ for a single
player:
\begin{definition}
Define 
four cones
 \begin{eqnarray*}
C_{++}(\vec{\theta}) &\equiv& \{\delta \vec{\theta} : \langle \grad(V), \delta \vec{\theta}\rangle   >  0,  \langle \grad(f), \delta \vec{\theta}\rangle >0\} \\
C_{+-}(\vec{\theta}) &\equiv& \{\delta \vec{\theta} : \langle \grad(V), \delta \vec{\theta}\rangle   >  0,  \langle \grad(f), \delta \vec{\theta}\rangle <0\} \\
C_{-+}(\vec{\theta}) &\equiv& \{\delta \vec{\theta} : \langle \grad(V), \delta \vec{\theta}\rangle   <  0,  \langle \grad(f), \delta \vec{\theta}\rangle >0\} \\
C_{--}(\vec{\theta}) &\equiv& \{\delta \vec{\theta} : \langle \grad(V), \delta \vec{\theta}\rangle   <  0,  \langle \grad(f), \delta \vec{\theta}\rangle <0\} .
\end{eqnarray*}
and also define
\begin{eqnarray*}
C_\pm(\vec{\theta}) &\equiv & C_{+-}(\vec{\theta}) \cup C_{-+}(\vec{\theta}).
\end{eqnarray*}
\end{definition}
\noindent 
So there are two hyperplanes, $\{\delta \vec{\theta} :  \langle \grad(V), \delta \vec{\theta}\rangle= 0\}$
and $ \{\delta \vec{\theta} :  \langle \grad(f), \delta \vec{\theta}\rangle= 0\}$, 
that separate the tangent space at $\vec{\theta}$ into the four disjoint convex cones $C_{++}(\vec{\theta})$,
$C_{+-}(\vec{\theta})$, $C_{-+}(\vec{\theta})$, $C_{--}(\vec{\theta})$. 
These cones are convex and pointed. In fact, each of them is contained in some open halfspace. 

By the definition of
the differential value of $f$ in the direction $\delta \vec{\theta}$, it is negative
for all $\delta \vec{\theta}$ in either $C_{+-}(\vec{\theta})$ or $C_{-+}(\vec{\theta}) = -C_{+-}(\vec{\theta})$, that is, in $C_\pm(\vec{\theta})$.

\subsection{Geometry of negative value of information}

In principle, either the pair of cones $C_{++}$ and $C_{--}$ or the
pair of cones $C_{+-}$ and $C_{-+}$ could be empty.  That would mean
that either \emph{all} directions $\delta \vec{\theta}$ have positive
value of $f$, or all have negative value of $f$, respectively. We now
observe that the latter pair of cones is nonempty --- so there are
directions $\delta \vec{\theta}$ in which the value of $f$ is negative
--- iff the gradient of $V$ and $f$ are not parallel.

\begin{proposition}
\label{prop_nonalign}
Assume that $\grad(V)$ and $\grad(f)$ are both
nonzero at $\vec{\theta}$. Then $C_{+-}(\vec{\theta)}$ and $C_{-+}(\vec{\theta)}$ are nonempty iff
\begin{equation}\label{2011}
  \langle \grad V(\vec{\theta}), \grad f(\vec{\theta}) \rangle  < |\grad V(\vec{\theta})| |\grad f(\vec{\theta})|
 \end{equation}
 \label{prop:neg_val_info}
\end{proposition} 

\begin{proof}
By the Cauchy-Schwarz inequality $\langle \grad V(\vec{\theta}), \grad f(\vec{\theta}) \rangle  \leq |\grad V(\vec{\theta})| |\grad f(\vec{\theta})|$
 with equality iff and only if $\grad V$ and $\grad f$ are positive multiples o each other. Thus, Eq. \rf{2011} means that
the two vectors $\grad(V)$ and $\grad(f)$ are not positively collinear.
It follows from  Lemma \ref{lemma:conic} (see appendix~\ref{app:ConicHull}) that  two (nonzero) vectors $\vec{v}_1,\vec{v}_2$ are not positively collinear iff
$Con(\{\vec{v}_1, \vec{v}_2\})$ is pointed, i.e., iff there are points in neither $Con(\{\vec{v}_1, \vec{v}_2\})$ nor its dual. That
in turn is equivalent to there being a third vector $\vec{w}$ with
\be
\langle \vec{v}_1, \vec{w}\rangle >0, \quad \langle \vec{v}_2, \vec{w}\rangle <0.
\qe
With $\vec{v}_1 =\grad(V), \vec{v}_2 =\grad(f)$, this means that Eq. \rf{2011} implies
that $C_{+-}\neq \emptyset$, and therefore
 $C_{-+}=-C_{+-}\neq \emptyset$. 
\end{proof}
\noindent We emphasize that this result (and other results below) are not predicated on our
use of the QRE or an information-theoretic definition of $f$. It holds even
for other choices of the solution concept and / or definition of ``amount of information" $f$. In addition, the 
requirement in Prop.~\ref{prop_nonalign} that $\vec{\theta}$ be in the interior of $\Theta$ is actually quite weak. This is because often
if  a given MAID of interest is represented by a $\theta$ on the border of $\Theta$ in one parametrization of the set of MAIDs, 
under a different parameterization 
the exact same MAID will correspond to a parameter $\theta$ in the interior of $\Theta$.

To illustrate Prop.~\ref{prop:neg_val_info}, consider a situation
where $\grad V$ and $\grad f$ are not parallel at $\vec{\theta}$, so
that $\grad (f) \not \propto \grad (V)$.  Suppose now, the player is
allowed to add any vector to the current $\vec{\theta}$ that has a
given (infinitesimal) magnitude. Then, in order to increase utility,
she would not choose the added infinitesimal vector to be parallel to
$\grad (f)$, i.e., she would prefer to use some of that added vector
to improve other aspects of the game's parameter vector besides
increasing $f$. Intuitively, so long as they value anything other than
$f$, we would generically expect there to be directions that have
negative value of $f$.  This holds for any player individually,
independent of possible relations, e.g. zero-sum, between the
utilities of different player.  Symmetrically, we would expect
directions that have positive value of $f$:

\begin{corollary}
\label{corr:always_good_directions}
Assuming that $\grad(V)$ and $\grad(f)$ are both
nonzero at $\vec{\theta}$, 
\begin{equation}\label{201}
  | \langle \grad V(\vec{\theta}), \grad f(\vec{\theta}) \rangle | < |\grad V(\vec{\theta})| |\grad f(\vec{\theta})|
 \end{equation}
implies that $C_{++}(\vec{\theta)}$ and $C_{--}(\vec{\theta)}$ are both nonempty.
\end{corollary}

\begin{proof}
Define $g(\vec{\theta}) \equiv -f(\vec{\theta})$, and write $C^g$ or $C^f$ to indicate whether
we are considering spaces defined by $V$ and $g$ or by $V$ and $f$, respectively.
\begin{eqnarray*}
  | \langle \grad V(\vec{\theta}), \grad f(\vec{\theta}) \rangle | & < & |\grad V(\vec{\theta})| |\grad f(\vec{\theta})| \\
\Leftrightarrow | \langle \grad V(\vec{\theta}), \grad g(\vec{\theta}) \rangle | & < & |\grad V(\vec{\theta})| |\grad g(\vec{\theta})| \\
\Rightarrow \langle \grad V(\vec{\theta}), \grad g(\vec{\theta}) \rangle & < & |\grad V(\vec{\theta})| |\grad g(\vec{\theta})| \\
\Leftrightarrow \quad C^g_{-+}(\vec{\theta)} &\mbox{ and }& C^g_{+-}(\vec{\theta)} \mbox{   are both non-empty} \\
\Leftrightarrow \quad C^f_{--}(\vec{\theta)} &\mbox{ and }& C^f_{++}(\vec{\theta)} \mbox{   are both non-empty} 
\end{eqnarray*}
where Prop.~\ref{prop:neg_val_info} is used to establish the second to last equality.
\end{proof}
\noindent Note that the converse to
Coroll.~\ref{corr:always_good_directions} does not hold. A simple
counter-example is where $\grad(V)(\vec{\theta}) = \grad
f(\vec{\theta})$, for which both $C_{++}(\vec{\theta)}$ and
$C_{--}(\vec{\theta)}$ are nonempty, but $| \langle \grad
V(\vec{\theta}), \grad f(\vec{\theta}) \rangle | = |\grad
V(\vec{\theta})| |\grad f(\vec{\theta})|$

Coroll.~\ref{corr:always_good_directions} means that so long as 
$| \langle \grad V(\vec{\theta}), \grad f(\vec{\theta}) \rangle | < |\grad V(\vec{\theta})| |\grad f(\vec{\theta})|$, there are
directions $\delta \vec{\theta}$ with positive value of $f$. This 
has interesting implications
for the analysis of value of information in the Braess' paradox and Cournot example of negative value
of information, as discussed below in Sec.~\ref{sec:more_examples}.  It 
also means that even if for some particular $f$ of interest one intuitively expects that increasing $f$ should reduce expected utility, generically
there will be infinitesimal changes to $\vec{\theta}$ that will increase $f$ but increase expected utility.
An illustration of this is also discussed in the ``negative value of utility" example in Sec.~\ref{sec:more_examples}.

\subsection{Genericity of negative value of information}

Consider situations where $f$ is a monotonically  increasing function of $V$ across a
compact $S \subseteq \Theta$. This means that  $f$
and $V$ have the same level hypersurfaces   across $S$
(although, of course the values of $V$ and $f$ on any such common level hypersurface will in general be different). 
The monotonicity implies that the  linear order
induced by values $f(\vec{\theta})$ relates the level hypersurfaces in the same way
that the linear order induced by values $V(\vec{\theta})$ relates those
level hypersurfaces. Say that in addition neither $\grad(f)$ nor $\grad(V)$ equals  0 anywhere
in $S$.  
So $\grad(V)$ and $\grad(f)$ are proportional to
one another throughout $S$ (although the proportionality constant may
change).{\footnote{Hypothesize that the gradients were not parallel
at some $\vec{\theta} \in S$. Since neither gradient equals zero, this would
mean that there is a direction $\delta \vec{\theta} \in S$ that has zero inner product
with $\grad f$ at $\vec{\theta$ but not with $\grad(V)$ there.  $\delta \vec{\theta}$ would lie in the level surface of $f$ going through $\vec{\theta}$
but not the level surface of $V$ going through $\vec{\theta}$. This would
contradict  the supposition that those level surfaces are coincident throughout $S$.}} 
So, by Prop.~\ref{prop:neg_val_info}, for no $\vec{\theta}$ in $S$
is there a direction $\delta \vec{\theta}$ such that $\set{V}_{f,\delta
  \vec{\theta}}(\vec{\theta}) < 0$.

In general though, level hypersurfaces will not match up throughout a region, as the condition that $\grad(V)$ and $\grad(f)$ be proportional is very restrictive and special, and so is typically violated.  When they do not, we
have points in that region that have directions with negative value of
$f$.  We now derive a criterion involving both  the gradients and the Hessians of $V$ and $f$ to identify such a mismatch.

\begin{proposition}
\label{prop_nonalign_generic}
Assume that both $f$ and $V$ are analytic at $\vec{\theta}$ with nonzero gradients, and
choose some $\epsilon >0$. Define
\begin{eqnarray*}
\nu &\equiv& \grad V(\vec{\theta}) \\
\phi &\equiv& \grad f(\vec{\theta}) \\
B_\epsilon(\vec{\theta}) &\equiv& \{\sigma: {\mathrm{dist}}(\sigma,\vec{\theta})< \epsilon\}
\end{eqnarray*}
Then
\begin{eqnarray}
\label{204}
C_{+-}(\vec{\theta}') =\; \varnothing && \!\!\!\!\!\!   \forall \; \vec{\theta}' \in B_\epsilon(\vec{\theta}) \\
\label{205}
\Rightarrow \quad \nu \propto \phi
&\mbox{and}& \quad \bigg|  \nabla \nabla(f) - \frac{|\phi|}{|\nu|} \nabla \nabla(V)  \bigg|  \frac{  \nu^i \nu_j}{ |\nu|^2} =    [\nabla \nabla(f)]^i_{j} - \frac{|\phi|}{|\nu|} [\nabla \nabla(V)]^i_j  \nonumber \\
\end{eqnarray}
where the proportionality constant in Eq.~\eqref{205} is greater than zero.
\label{lemma:problems_everywhere}
\end{proposition}

\begin{proof}
We start with an observation. Whereas the two vectors $\grad(V)$ and $\grad(f)$ depend on the metric, the question of whether
those  two vectors are collinear does not depend on the choice of metric. (This is in contrast to the question of whether they are orthogonal, which needs a metric.) However as these two vectors depend on the metric, and as we shall have to take derivatives and therefore cannot stay in a single tangent space, we shall have to employ some of the differential geometry concepts reviewed in the appendix.

Hypothesize that $C_{\pm}(\vec{\theta}') = \varnothing$. From Prop.~\ref{prop:neg_val_info}
this is true iff $\grad(V)(\vec{\theta}')$ and $\grad(f)(\vec{\theta}')$ are positively proportional, that is,
\bel{205b}
\grad(f)(\vec{\theta}')=\rho(\vec{\theta}')\grad(V)(\vec{\theta}')
\qe
for some positive function $\rho$. So using the definitions in the proposition, 
$\rho(\vec{\theta})= \frac{|\phi|}{|\nu|}$.
In addition, if \rf{204} holds,
we can lower both the contravariant vector field $\grad f$ and the contravariant vector field $\grad(V)$
in \rf{205b}  into covariant vector fields, and then apply the covariant derivative operator 
with respect to $\theta'$ to both sides.
If we then raise all lower indices back up to upper indices, we derive
\bel{207}
[\nabla \nabla f(\theta')]^{i,j}- \rho(\theta') [\nabla \nabla V(\theta')]^{i,j} =[\grad \rho(\theta')]^i [\grad(V)(\theta')]^j
\qe

Since both Hessians in this equation are the covariant Hessians, and therefore symmetric with respect to $i,j$, we conclude by symmetrization that 
\bel{209}
[\grad \rho(\theta')]^i [\grad(V)(\theta')]^j = [\grad \rho(\theta')]^j [\grad(V)(\theta')]^i
\qe
for all pairs $(i, j)$. So in particular, evaluating all terms at the center of the ball, $\theta' = \theta$, we see that
\be
\frac{[\grad \rho]^j} {[\grad(V)]^j}   = \frac{[\grad \rho]^i} {[\grad(V)]^i}
\qe
for all $i, j$ such that $[\grad(V)]^i \ne 0$ and $[\grad(V)]^j \ne 0$. Accordingly $[\grad \rho]^i \propto [\grad(V)]^i$
for all $i$ such that $[\grad(V)]^i \ne 0$ where the proportionality constant is independent of $i$. 
So evaluating at $\theta' = \theta$ and by the above proportionality, \rf{207} simplifies to
\bel{209b}
[\nabla \nabla f]^{i,j}- \frac{|\phi|}{|\nu|} [\nabla \nabla V]^{i,j} = k [\grad(V)]^i [\grad(V)]^j
\qe
for some scalar $k$. Taking norms of both sides, yields
\bel{210}
k=\frac{ \bigg|  \nabla \nabla f - \frac{|\phi|}{|\nu|} \nabla \nabla V  \bigg|} { |\nu|^2},
\qe
whence \rf{205}. 

\end{proof}

\noindent We will refer to the matrix with entries 
\begin{eqnarray*}
\bigg|  \nabla \nabla  f - \frac{|\phi|}{|\nu|} \nabla \nabla V  \bigg|  \frac{  \nu^i \nu_j}{ |\nu|^2} -  [\nabla \nabla f]^{i}_j + \frac{|\phi|}{|\nu|} [\nabla \nabla f]^{i}_jV 
\end{eqnarray*}
as the second-order mismatch.
{\footnote{
One might think that so long as the Hessian of $V$ at $\vec{\theta}$ is not a
constant times the Hessian of $f$ at $\vec{\theta}$, then the level curves of $V$ and $f$ 
will not match up throughout the vicinity of $\vec{\theta}$. Such mismatch in the level curves would
in turn imply that even if the gradients of $V$ and $f$ were parallel at $\theta$, there would be points infinitesimally close to $\vec{\theta}$
where the gradients of $V$ and of $f$ are not parallel. If this reasoning were valid, it would suggest
that the condition in Prop.~\ref{lemma:problems_everywhere} that the second-order mismatch is non-zero could
be tightened (by replacing that condition with the requirement that $D_{ij}F  \ne \frac{|\phi|}{|\nu|} D_{ij}V$) 
and yet the
implication given by that proposition --- that there are points where the gradients are not parallel --- would still follow.
This is not the case though: Even if the Hessian of $V$ is not a
constant times the Hessian of $f$ at $\vec{\theta}$, it is possible that the gradients of $f$
and $V$ are parallel throughout a neighborhood of $\vec{\theta}$. As a simple example, take $d = 2, V(\vec{\theta}) = \theta^1,
f(\vec{\theta}) = (1 + [\theta^1])^2$. Even though the Hessian of this $V$ is not a
constant times the Hessian of this $f$ anywhere near the origin, the gradients of $f$ and $V$ are
parallel everywhere near the origin. }}
Prop.~\ref{lemma:problems_everywhere} means that if there is nonzero second-order mismatch at some $\vec{\theta}$
then $C_{+-}(\vec{\theta})$ cannot be empty throughout a $d$-dimensional sphere centered at $\vec{\theta}$.
In other words, it means that for any $\epsilon > 0$, there is some $\vec{\theta}'$ within distance $\epsilon$ of
$\vec{\theta}$ and associated direction $\delta \vec{\theta}'$ such that at $\vec{\theta'}$ there are negative values of $f$ in direction $\delta \vec{\theta}'$.{\footnote{
However it is still possible for $C_{+-}(\vec{\theta})$ to be empty throughout a $(d-1)$-dimensional hyperplane $T$
going through $\vec{\theta}$. 
As an example, choose $d = 2, V(\vec{\theta}) = \theta^1, 
f(\vec{\theta}) = (1 + \theta^1)^2 (1 + [\theta^2]^2)$. The second order mismatch is nonzero for this $V$ and $f$
at $\vec{\theta} = (0, 0)$. (The Hessian of $f$ at the origin is twice the identity matrix, the Hessian of $V$ there
is the zero matrix, but $\nu_i \nu_j$ only has one non-zero entry, for $i = j = 1$.)  However everywhere on the line 
$\{\vec{\theta}' = t(1, 0) : t \in \R\}$ going through the origin, $\grad f(\vec{\theta}')$ and $\grad(V)(\vec{\theta}') $ are parallel.}}

In the light of this, arguably the most important issues to analyze in any
given game is  not whether at some
particular $\vec{\theta}$ there is a direction in which some particular type of information increases
but expected utility decreases. That is almost always the case. 
Rather our results redirect attention to different issues.  One such
issue is what fraction of all $\vec{\theta}$ and all possible infinitesimal
changes to $\vec{\theta}$ have negative value of information.

Another important issue arises whenever there are constraints on the
allowed directions in $\Theta$ specifying how $\vec{\theta}$ is changed. Such constraints are quite common in
the real world. As an example, there might be inequality constraints on the
information capacity of a channel, which would manifest themselves in
constraints on the allowed change in parameter space for $\vec{\theta}$ for
which those constraints are tight. Another example
of a constraint on the allowed change to the parameter vector is the stipulation
that the change must be an insertion (or removal) of a garbling node, in the
sense of Blackwell's theorem~\cite{leshno1992elementary}.
With such constraints, in general  the directions in the sets $C_{+-}(\vec{\theta)}$ and $C_{-+}(\vec{\theta)}$ arising in Prop.~\ref{prop:neg_val_info} 
may all violate the constraints, and so not be allowed.
This is why,
for example, the result of Blackwell (that there are \emph{no} garbling
nodes that one can insert just before a decision node in an ID that
increase maximal expected utility for a rational player) is consistent with our results;
the insertion of such a garbling node corresponds to a constrained
change in $\vec{\theta}$, a change that lies in a cone that does not
intersect $C_\pm(\vec{\theta})$ (see Fig.~\ref{fig:BlackwellIso}).

\subsection{Pareto negative values of information}

Recall that Prop.~\ref{prop:neg_val_info} established that 
there are directions in which the value of $f$ is negative for player $i$ iff 
$\grad f$ and $\grad V_i$ are not positively proportional.  We can generalize this to
give a condition for there to be directions in which the value of $f$
is negative for \emph{all} players, in which case we say that we have ``Pareto negative values
of information".
To make this precise, define $\set{V}^i_{f, \delta \vec{\theta}} (\vec{\theta})$ as the value
of $f$ in direction $\delta \vec{\theta}$ to player $i$ in a MAID.

\begin{proposition}
\label{prop_allneg}

Given a set of functions $V_i : \Theta \rightarrow \R$, define $C \equiv Con(\{\grad(V_i)\})$.
Assume $C^\bot \ne \varnothing$. Then
\begin{eqnarray}
\nonumber
\grad(f) &\notin&  C  \nonumber \\
\Leftrightarrow \quad \exists \; \delta {\bftheta} &:&  \mathscr{V}^i_{f, \delta {\bftheta}} ({\bftheta}') < 0 \; \forall i
\end{eqnarray}
where all $N + 1$ gradients are evaluated at ${\bftheta}$. 
\label{prop:pareto_neg_val_f}
\end{proposition}

\begin{proof}
Note that the dual cone $Con(\grad f)^\bot = H_-(\grad f)$. So by Lemma~\ref{lemma:conic_hull_duals} (see appendix~\ref{app:ConicHull}), 
since $C^\bot$ is nonempty, $C^\bot \subseteq H_-(\grad(f))$ iff $\grad(f) \in C$.
Thus,  $\grad(f) \notin C$ means that the set difference $C^\bot \setminus H_-(\grad(f))$ is nonempty. Since those
two sets are defined by strict inequalities, this set difference is open. So in fact $\grad(f) \notin C$
 iff there is some $\delta {\theta}$ such that both $\delta {\theta} \in C^\bot$ and $\delta {\theta} \notin cl [H_-(\grad(f))]$,
 the closure of $H_-(\grad(f))$. For that $\delta {\theta}$, $\langle \delta {\theta}, V_i\rangle <0$ for all $i$, but $\langle \delta {\theta} ,\grad(f)\rangle >0$. This means $\mathscr{V}^i_{f, \delta {\theta}} ({\theta}) < 0$ for all $i$. This argument also works in the converse direction. 
\end{proof}

\noindent The existence condition in Prop.~\ref{prop:pareto_neg_val_f}
is crucial. To give a simple example, note that if $\grad V_1 = -\grad
V_2$, then even if $\grad f \notin C$, there is no direction $\delta
\vec{\theta}$ such that both $\set{V}^1_{f, \delta \vec{\theta}}
(\vec{\theta}) < 0$ and $\set{V}^2_{f, \delta \vec{\theta}}
(\vec{\theta}) < 0$.{\footnote{To see this, note that if there were
    such a $\delta \vec{\theta}$, it could not lie in
    $H_0(\grad(V)^1)$, and so must either lie in $H_-(\grad(V)^1)$ or
    $H_+(\grad(V)^1) = H_-(\grad(V)^2)$. Wolog assume the former. Then
    $\set{V}^1_{f, \delta \vec{\theta}} (\vec{\theta}) < 0$ would
    require that $\delta \vec{\theta} \in H_+(\grad f)$.  However
    $\delta \vec{\theta} \in H_-(\grad(V)^1) = H_+(\grad(V)^2)$ would
    then force $\set{V}^2_{f, \delta \vec{\theta}} (\vec{\theta}) >
    0$.}}  The reason this is not a counterexample to
Prop.~\ref{prop:pareto_neg_val_f} is that $C^\bot$ is empty in this
case.  An simple example where such a situtation arises would be a
constant-sum game, i.e. $\sum_i u_i(\vec{\theta}) = c$ for all
$\vec{\theta}$\footnote{Note that if the sum of the utilities of the
  players depends on $\vec{\theta}$, $Con(\{\grad(V_i)\})$ could be
  pointed even though the game is constant sum for any particular
  $\vec{\theta}$ (albeit with a different constant). Thus, The
  character of games infinitesimally close to $\vec{\theta}$ is what
  is important.}. Then, $Con(\{\grad(V_i)\})$ is not pointed since
$\grad V_1 = - \sum_{i=2}^N \grad V_i$ and so, by
Lemma~\ref{lemma:conic_props}, $C^\bot = \varnothing$.

When the condition in Prop.~\ref{prop:pareto_neg_val_f} holds,
there is a direction from $\vec{\theta}$ that is Pareto-improving, yet
reduces $f$. So for example, if that condition holds for a particular
set of MAIDs when $f$ is the information capacity of a channel in the
MAID, then all players would benefit from reducing the capacity of
that channel.

Similarly, to the second order mismatch condition above which implies
that negative value of information is rather generic, the conditions
of Prop.~\ref{prop:pareto_neg_val_f} are expected to hold if the
number of players $N$ is less than the dimension $d$ of the parameter
space $\Theta$. Simply observe that $Con(\{\grad(V_i)\})$ is spanned
by $N$ vectors and thus of at most dimension $N < d$. Now, it is
rather non-generic that an unrelated vector, e.g. $\grad(f)$, happens
to lie in this lower-dimensional subspace. Thus, in this case, Pareto
improving changes to $\theta$ which decrease the value of $f$ should
be prevalent.\footnote{However it is important to remember that Pareto-improving changes to
$\vec{\theta}$ are not the same thing as Pareto-improving changes to a
strategy profile in a game specified by a single $\vec{\theta}$.  We can
have Pareto-improving changes to $\vec{\theta}$ even if $\vec{\theta}$ specifies a
game in which there are no Pareto-improving changes to the equilibrium
strategy profile.}

Note that the bit of whether or not $\grad(f) \in Con(\{\grad(V_i)\}$
is a covariant quantity. So the implications of
Prop.~\ref{prop:pareto_neg_val_f} do not change if we change the
coordinate system.  In fact, the value of that bit is independent of
our choice of the metric, so long as no $\grad(V_i)$ is in the kernel
of the metric. 

A similar result, can be obtained for Pareto positive value of
information.  By Def.~\ref{def:val_f_in_direction}, for any $i$, $
\mathscr{V}^i_{-f, \delta {\theta}} ({\theta}') < 0$ iff $
\mathscr{V}^i_{f, \delta {\theta}} ({\theta}') > 0$.  So an immediate
corollary of Prop.~\ref{prop:pareto_neg_val_f} is that that whenever
$C^\bot \ne \varnothing$, there a direction $\delta \theta$ in which $
\mathscr{V}^i_{f, \delta {\theta}} ({\theta}') > 0 \; \forall i$ iff
$-\grad(f) \;\notin\; C$.  So if $C^\bot \ne \varnothing$, then if in
addition the \emph{negative} of the gradient of $f$ is not contained
in the conic hull of the gradients of the players' expected utilities,
there is a direction in which we can change the game parameter vector
which will increase both $f$ and expected utility for all players.

\section{Examples of differential value}
\label{sec:Examples}
\label{sec:two_player}

\subsection{Decision problem}

As a first example to illustrate our approach we consider a game
involving a single player. In such a decision problem one agent plays against Nature. (So this MAID is
the special case of an influence diagram.) A simple example, which we
will use to illustrate our notion of differential value of information
is shown in \fig{decision}.

\showfigure{
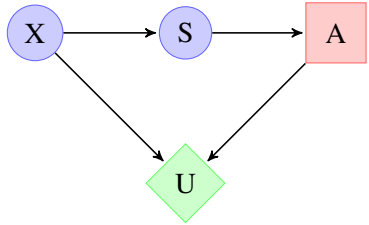
\begin{figure}[h]
  \begin{center}
    \begin{tikzpicture}
      \tikzstyle{nature}=[draw=blue!60,fill=blue!20,circle,minimum size=5mm]
      \tikzstyle{player}=[draw=red!60,fill=red!20,rectangle,minimum size=8mm]
      \tikzstyle{utility}=[draw=green!60,fill=green!20,diamond,minimum size=8mm]
      \tikzstyle{standard}=[shorten >=1pt,>=stealth',semithick]
      \node (nature) at ( 0,0) [nature] {X};
      \node (channel) at ( 2,0) [nature] {S};
      \node (player) at ( 4,0) [player] {A};
      \node (utility) at (2,-2) [utility] {U};
      \draw [->] (nature) edge[standard] (channel);
      \draw [->] (channel) edge[standard] (player);
      \draw [->] (nature) edge[standard] (utility);
      \draw [->] (player) edge[standard] (utility);
    \end{tikzpicture}
  \end{center}
  \caption{\label{fig:decision} A simple decision situation.}
\end{figure}}

In this MAID there is a state of Nature random variable $X$ taking on
values $x$ according to a distribution $p(x)$.  The agent observes $x$
indirectly through a noisy channel that produces a value $s$ of a
signal random variable $S$ according to a probability distribution
$p(s \mid x;\bftheta)$ parametrized by $\vec{\theta}$. The agent then takes
an action, which we write as the value $a$ of the random variable $A$,
according to the distribution $\sigma(a \mid s)$.  Finally, the
utility $u(x,a)$ is a function that depends only on $x$ and $a$.

Here, we consider a simple setup using a
binary state of Nature $p(X = 0) = p(X = 1) = \frac{1}{2}$ and a
binary channel $p(s | x; \bftheta)$ with parameters $\vec{\theta} =
(\epsilon^1, \epsilon^2)$ specifying the transmission errors for the
two inputs
\[ p(s | x; \vec{\theta}) = \left\{ \begin{array}{ll}
    1 - \epsilon^1 & \mbox{if } x = 0, s = 0 \\
    \epsilon^1 & \mbox{if } x = 0, s = 1 \\
    \epsilon^2 & \mbox{if } x = 1, s = 0 \\
    1 - \epsilon^2 & \mbox{if } x = 1, s = 1
    \end{array} \right.
  \] 

The player has two moves and obtains utility
\[
u(x, a) = \left\{ \begin{array}{ll}
    0  & x = 0, a = 0 \\
    -2 & x = 0, a = 1 \\
    0  & x = 1, a = 0 \\
    1 & x = 1, a = 1
  \end{array} \right.
\]
Thus, the player has a safe action $A = 0$, but can obtain a higher
utility by playing $1$ when she is certain enough that the state of
Nature is $1$.

We will refer to this MAID as the \textbf{Blackwell ID},
since it corresponds to the situation analyzed by Blackwell.
A fully rational
decision maker will set $\sigma(a \mid s)$ to maximize the
expected utility. 
Here, to have a differentiable strategy, we assume that the agent is
not fully rational but plays a quantal best response with a finite
$\beta$:
\[ \sigma(a | s) = \frac{1}{Z(s)} e^{\beta \E(u \mid s, a ;
  \bftheta)} \] where $Z(s) = \sum_a e^{\beta \E(u \mid s, a ;
  \bftheta)}$ and $\E(u \mid s, a ; \bftheta)$ denotes the conditional
expected utility for a given signal $s$ and action $a$.

Fig.~\ref{fig:BlackwellIso} show the isoclines of the
mutual information $I(X ; S)$ and the expected utility $V(\beta,
\epsilon^1, \epsilon^2)$ for $\beta = 5$. Both mutual information and
expected utility improve with decreasing channel noise, i.e. with reducing
$\epsilon^1$ and/or $\epsilon^2$. Nevertheless, the isoclines of these quantities do not 
exactly match. Thus it is
possible to change the channel parameters such that the expected
utility increases while the mutual information decreases. As an
example, consider any parameter combination ($\epsilon_1, \epsilon_2$)
where the isoclines plotted in Fig.~\ref{fig:BlackwellIso} intersect. 
Moving into the region ``above'' the $MI$ isocline and below the $V$
isocline will increase $V$ while decreasing $MI$. 

This potential inconsistency between changes to mutual information and
changes to expected utility does  not violate Blackwell's
theorem. That is because we allow arbitrary changes to the channel parameters; those that
result in the inconsistency
which cannot be represented as garbling in the sense of
Blackwell. As an illustration, for one particular pair of game parameter values, the set of
channels which are more or less informative according to Blackwell's partial order are
visualized as the light and dark gray regions in fig.~\ref{fig:BlackwellIso} respectively. 

This potential for an inconsistency between changes to information and
utility is also illustrated by considering the gradient vectors $\grad
I(X; S)$ and $\grad(V)$ which are orthogonal to the isoclines (wrt/
the Fisher metric)\footnote{Consider a (joint) distribution $p(x ;
  \vec{\theta})$ that is parameterized by $\vec{\theta}$. Then, the
  Fisher information metric $g_{ij}$ is defined as
\[ g_{ij} = \sum_x p(x; \vec{\theta}) \frac{\partial \log p(x;
  \vec{\theta})}{\partial \theta^k} \frac{\partial \log p(x;
  \vec{\theta})}{\partial \theta^l} \] The statistical origin of the
Fisher metric lies in the task of estimating a probability
distribution from a family parametrized by $\vec{\theta}$ from
observations of the variable $x$. The Fisher metric expresses the
sensitivity of the dependence of the family on $\vec{\theta}$, that
is, how well observations of $x$ can discriminate among nearby values
of $\vec{\theta}$.}. By Prop.~\ref{prop_nonalign}, only if those
gradients are collinear is it the case that \emph{every} change to the
game parameter vector increasing the expected utility must necessarily
increase the mutual information. However in
Fig.~\ref{fig:BlackwellIso} we clearly see that these gradients giving
the directions of steepest ascent of $I(X; S)$ and $V$ are
different. Thus, locally we can always find directions with negative
value of information. Note that this does not mean that the agent
prefers less information globally, i.e. there is no path, along which
the value of information is strictly negative, from no information,
e.g. $\epsilon_1 = \epsilon_2 = \frac{1}{2}$ to full information
($\epsilon_1 = \epsilon_2 = 0$).

Conversely,
Cor.~\ref{corr:always_good_directions} implies that we can find
infinitesimal changes to the game parameters that cause both expected utility and mutual information
to increase, in agreement with Blackwell's theorem.  In
the present example, such changes arise if we simultaneously
reduce both channel noises, e.g. by moving directly towards the origin.

\showfigure{
\begin{figure}[h]
  \includegraphics[width=0.8\textwidth]{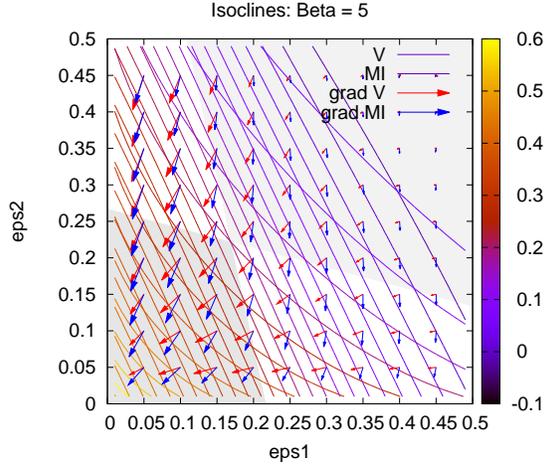}
  \caption{\label{fig:BlackwellIso} Isoclines of expected utility $V$
    and mutual information $I(X; S)$ with corresponding gradient
    vectors showing the directions of steepest ascent (wrt/ the Fisher
    metric).  The gray regions show which channels are more informative (dark
    gray) or less informative (light gray) than the channel with noise
    parameters $(\epsilon_1 = 0.17, \epsilon_2 = 0.22)$, in the sense of the term arising
    in Blackwell's analysis.}
  \end{figure}}

\subsection{Leader-follower example}

Next, we turn to a simple games involving two players, namely a leader $A^1$ and a
follower $A^2$ (see \fig{game}). 
\showfigure{
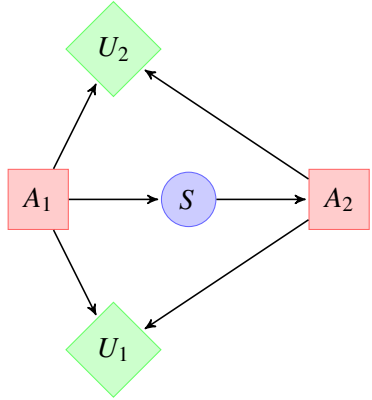
\begin{figure}[h]
  \begin{center}
    \begin{sideways}
      \begin{tikzpicture}
        \tikzstyle{nature}=[draw=blue!60,fill=blue!20,circle,minimum size=5mm]
        \tikzstyle{player}=[draw=red!60,fill=red!20,rectangle,minimum size=8mm]
        \tikzstyle{utility}=[draw=green!60,fill=green!20,diamond,minimum size=8mm]
        \tikzstyle{standard}=[shorten >=1pt,>=stealth',semithick]
        \node[rotate=-90] (player1) at ( 4,2) [player] {$A_1$};
        \node[rotate=-90] (channel1) at ( 4,0) [nature] {$S$};
        \node[rotate=-90] (player2) at ( 4,-2) [player] {$A_2$};
        \node[rotate=-90] (utility2) at (6,1) [utility] {$U_2$};
        \node[rotate=-90] (utility1) at (2,1) [utility] {$U_1$};
        \draw [->] (player1) edge[standard] (channel1);
        \draw [->] (channel1) edge[standard] (player2);
        \draw [->] (player1) edge[standard] (utility1);
        \draw [->] (player1) edge[standard] (utility2);
        \draw [->] (player2) edge[standard] (utility2);
        \draw [->] (player2) edge[standard] (utility1);
      \end{tikzpicture}
    \end{sideways}
  \end{center}
  \caption{\label{fig:game} A 2-player game where player $A^1$
    (leader) can move before player $A^2$ (follower).}
\end{figure}}
In contrast to the single-player game, now  the distribution
$p(X)$ of the state of Nature is replaced by the equilibrium strategy
$\sigma(a_1)$ of player $1$, a strategy that will also depend on the parameters
$\bftheta$. Another difference is that there are now two utility functions, one
for each player. 

As in the decision problem, we consider binary state spaces and an
asymmetric binary channel with parameters $\bftheta = (\epsilon^1,
\epsilon^2)$. We use the utility functions of the players analyzed in
\cite{bagw95}\footnote{The game is a discretization of the Stackelberg duopoly game, with two moves for each
  player.}  
\begin{center}
  \begin{tabular}{cc|c@{\hspace{8ex}}c@{\hspace{8ex}}}
    \multicolumn{2}{c}{} & \multicolumn{2}{c}{follower} \\
    & $u_1$/$u_2$ & $L$ & $R$  \\ \hline
    \multirow{2}{1em}{\begin{sideways}leader\end{sideways}} & $L$ & (5, 2) & (3, 1) \\
    & $R$ & (6, 3) & (4, 4)
  \end{tabular}
\end{center}
Bagwell pointed out that in the pure strategy Nash equilibria the leader can only
take advantage of moving first (by playing $L$), when the follower can
observe his move perfectly (Stackelberg solution). As soon as the
slightest amount on noise is added to the channel only the equilibrium of the
simultaneous move game (both playing $R$, the Cournot solution) remains\footnote{There are
  additional mixed equilibria, which change smoothly with the
  noise. These are mentioned  in   Bagwell \cite{bagw95}, but not discussed further.}.
Here we show that our differential analysis uncovers a much
richer structure. In particular, we show there exist a QRE branch and
parameters for the noise of the channel such that both players prefer
more noise.

In the decision case, we used $I(X; S)$ to quantify the amount of
information that is available to the (single) player. In the multi-player game setting, the
corresponding quantity is $I(A_1; S)$, and strongly depends on the move of
the leader. As an illustration, 
consider a symmetric channel $p(s | a_1)$ parametrized by the single value $\epsilon \equiv
\epsilon^{1,2}$. Fig.~\ref{fig:BagwellQREEpsBeta} shows how the strategy
of the leader depends on the channel noise
and the rationality $\beta = \beta^{1,2}$ of the players.
\showfigure{
\begin{figure}[h]
  \begin{center}
    \includegraphics[width=0.7\textwidth]{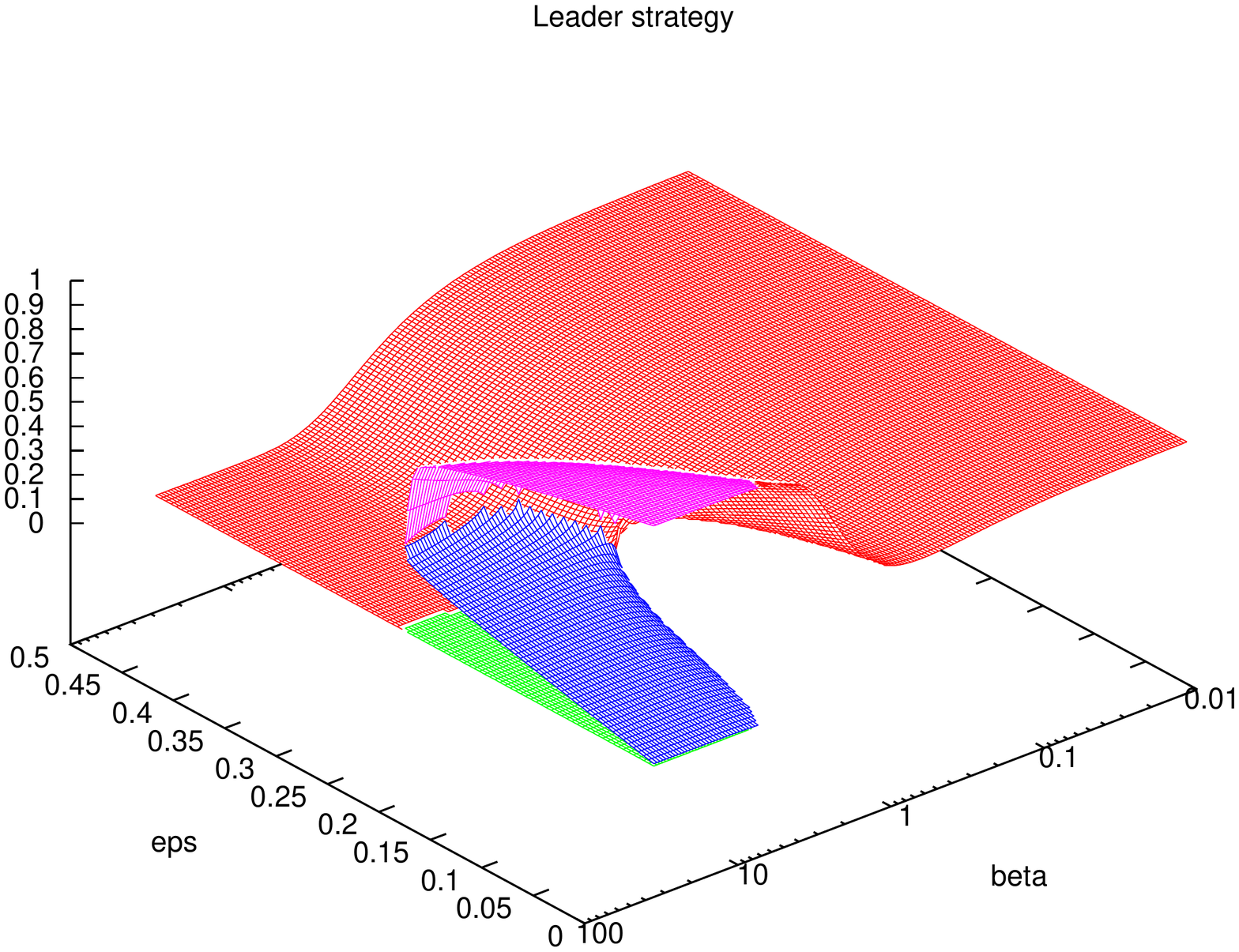}
  \end{center}
  \caption{\label{fig:BagwellQREEpsBeta} Surface of QRE equilibria
    for the symmetric channel leader-follower game. 
  }
\end{figure}}
This shows that for sufficiently rational players ($\beta > 5$) there exist
multiple QRE solutions. For $\beta \to \infty$ the three QRE
equilibria converge to the pure strategy Nash equilibrium where both
play $R$ (Cournot outcome, lower branch in red/green) and to the two
mixed strategy Nash equilibria of the original game, respectively.\footnote{This demonstrates that our
  analysis can easily be extended to analyze Nash equilibria. In this case,
  choosing a branch corresponds to choosing a particular equilibrium, and the
  partial derivatives $\frac{\partial \sigma}{\partial \epsilon^i}$
  vanish, as long as the equilibrium exists.} For $\epsilon = 0$, the
upper mixed strategy equilibrium coincides with the equilibrium mentioned above
where the leader has an advantage. In the following, we focus on the branch
that smoothly connects to the origin $\beta = 0$, the so-called
``principal branch'', which includes that upper equilibrium.

\showfigure{
\begin{figure}[h]
  \begin{center}
    \begin{tabular}{cc}
      \hspace*{-10ex} \includegraphics[width=0.6\textwidth]{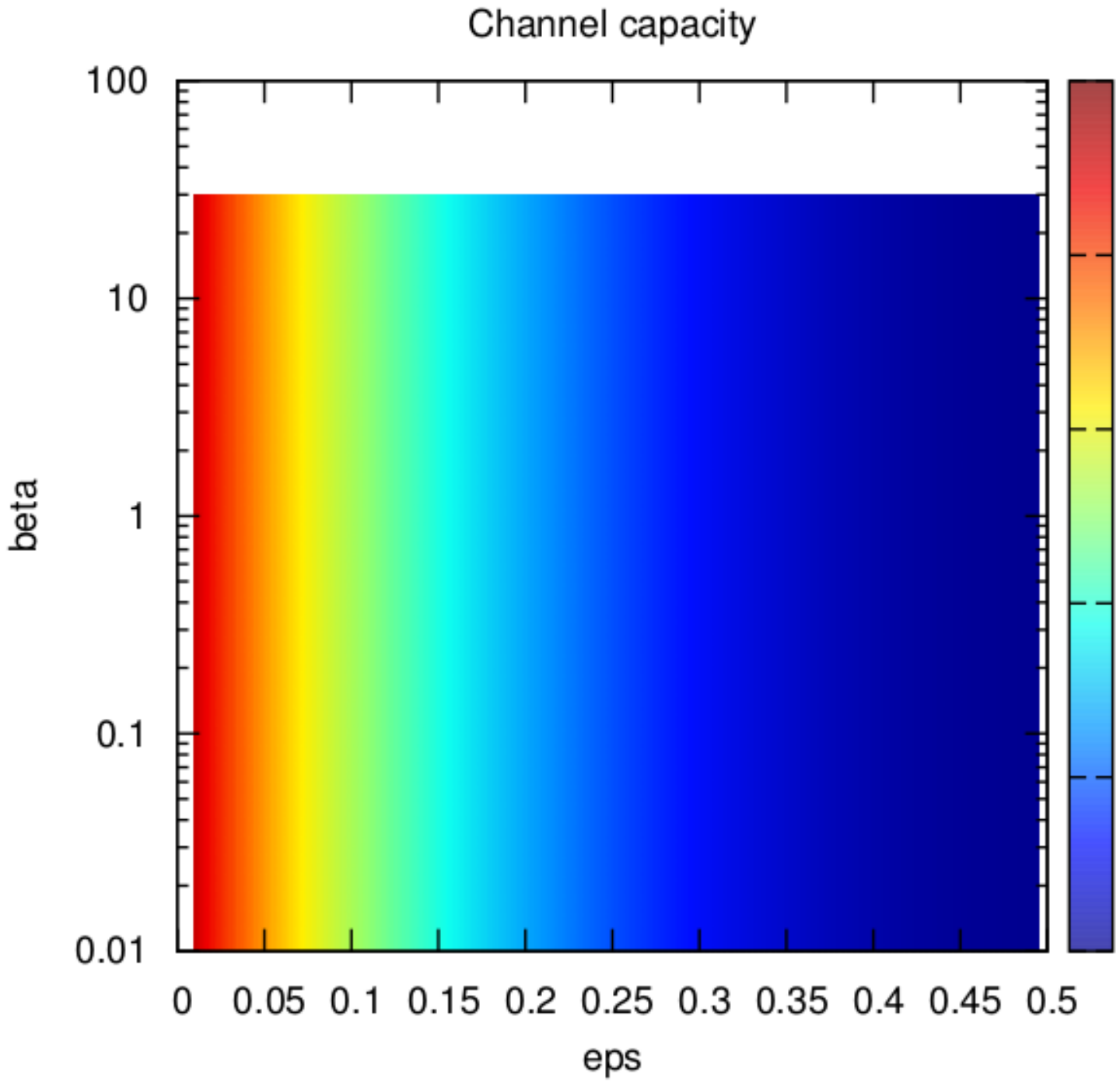} &
      \hspace*{-10ex} \includegraphics[width=0.6\textwidth]{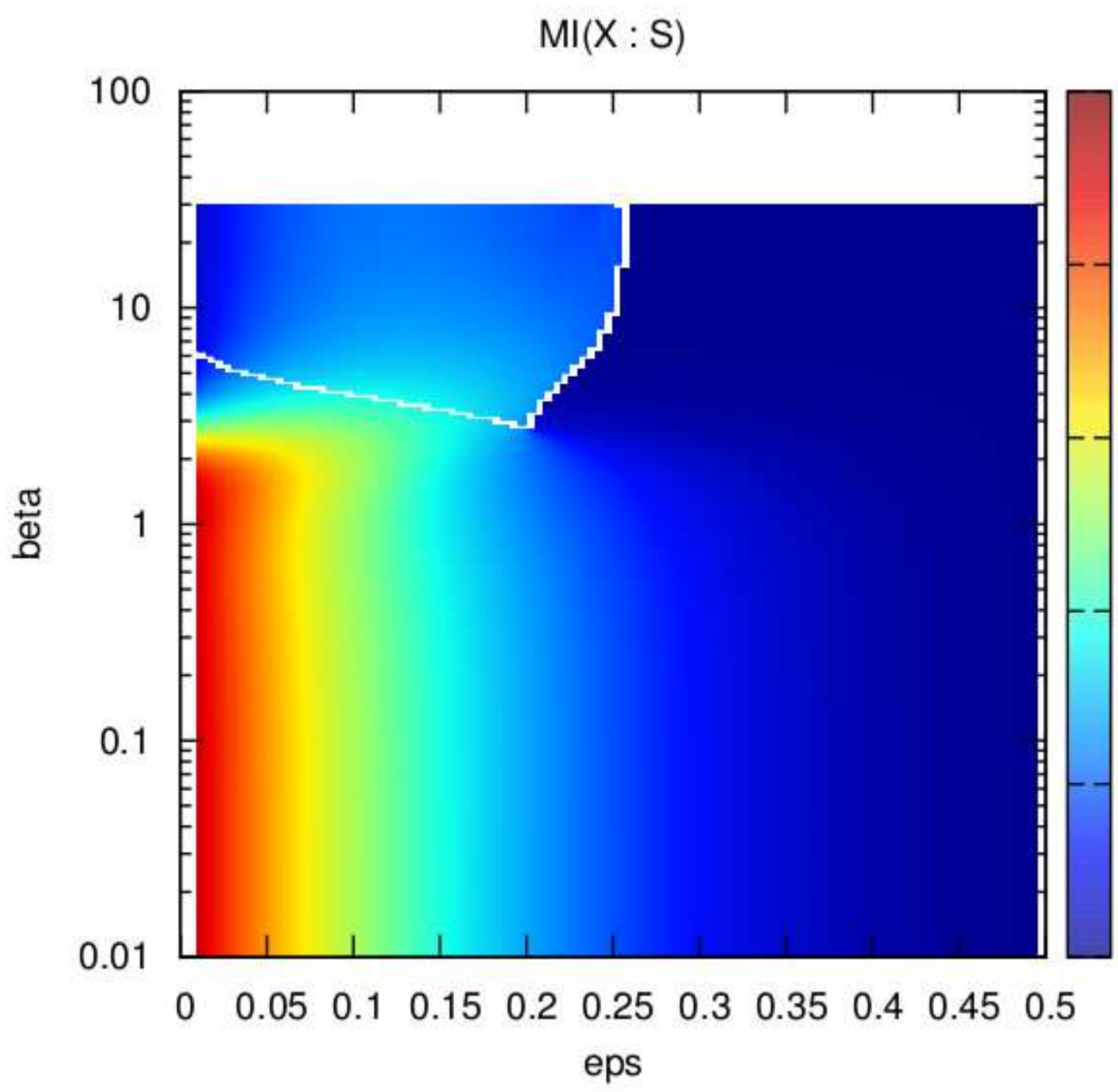}
    \end{tabular}
  \end{center}
  \caption{\label{fig:BagwellMIEpsBeta} Channel capacity and mutual
    information $I(A_1; S)$ in the symmetric channel leader-follower game.}
\end{figure}}

Fig~\ref{fig:BagwellMIEpsBeta} shows the channel capacity as a function of the conditional
distribution in the channel, as well as
the mutual information $I(A_1; S)$ that is actually transferred across
the channel. As soon as the leader is rational enough, he starts to
prefer the move $L$. This means that the mutual information $I(A_1; S)$
decreases when the leader gets rational enough.{\footnote{Remember that 
$I(A_1; S) = \E_{p(A_1)}[D_{KL}(p(S|a_1) \mid\mid p(S))]$ and thus it vanishes if the leader plays
a pure strategy, since the average becomes trivial.}} However, the potential information that could be
transferred, i.e., the channel capacity, is independent of player strategies, and so is still high.
This illustrates how 
studying the information capacity rather than the mutual information
is perhaps more in line with standard game theory, where the
information partition is considered as part of the specification of the game parameters,
independent of the resultant player strategies.
For these reasons we focus our analysis on the channel capacity.

In this simple symmetric-channel scenario the space of game parameters concerning the channel noise is one-dimensional, and so
analyses of ``gradients" over that space are not particularly illuminating. Accordingly,
to further investigate the role of information in the leader-follower
game we consider an asymmetric channel, so that $p(s | a_1)$ is
parametrized by two noise parameters, $\epsilon^1, \epsilon^2$,
giving the probability of error for the two inputs $a_1 = L$ and $a_1 = R$
respectively. We also fix $\beta = 10$ for both players. In this case, we
again find multiple QRE branches when the channel noise is
small enough.
We focus on the branch where
the leader has the biggest advantage and can achieve the highest
utility while the utility of the follower is lowest. In the following, 
we refer to this QRE solution as the ``Stackelberg branch''.

\showfigure{
}

\showfigure{
\begin{figure}[h]
  \begin{tabular}{ll}
    {\bf A)} & {\bf B)} \\[-5ex]
    \hspace*{-10ex} \includegraphics[width=0.6\textwidth]{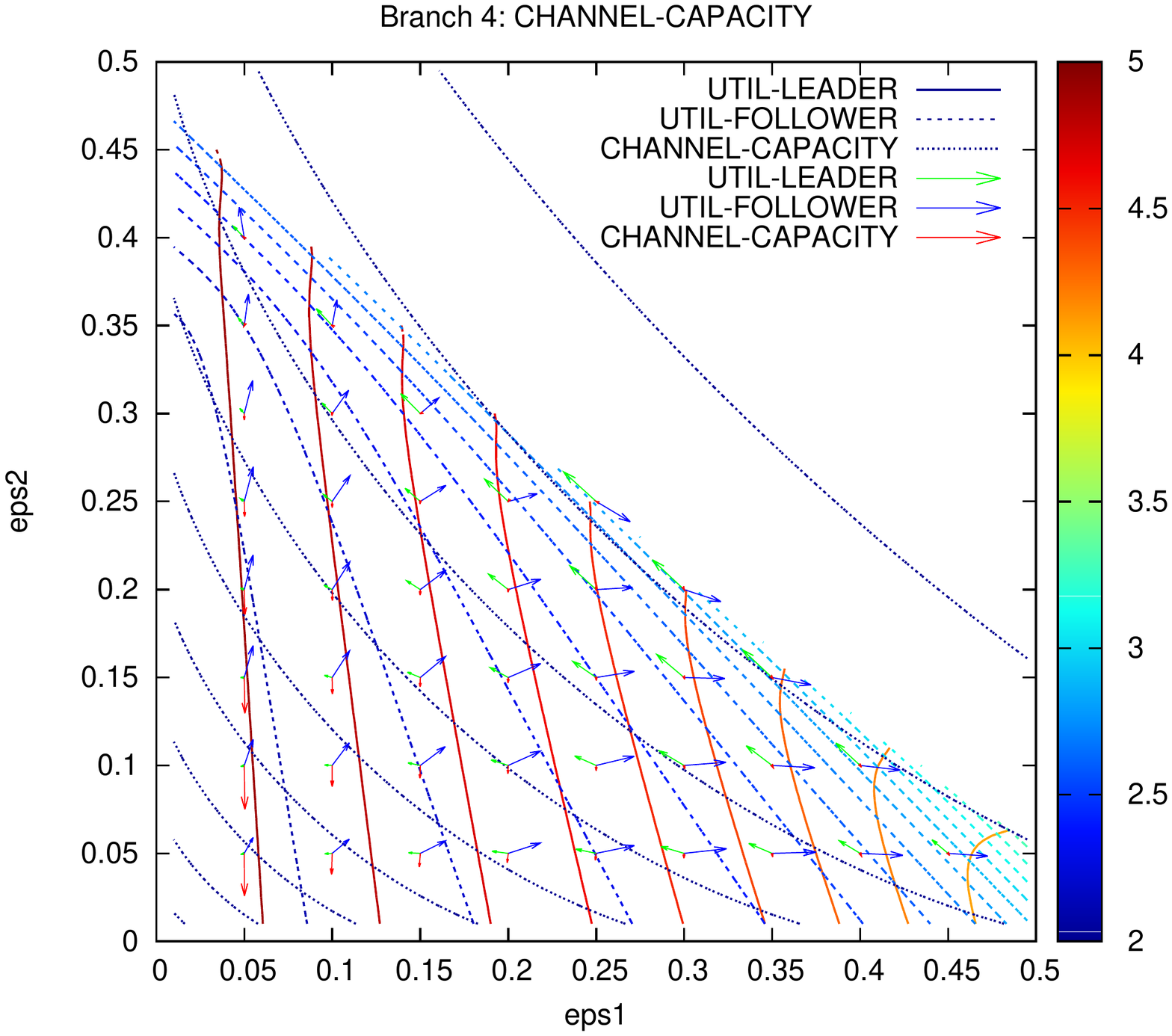} &
    \hspace*{-10ex} \includegraphics[width=0.6\textwidth]{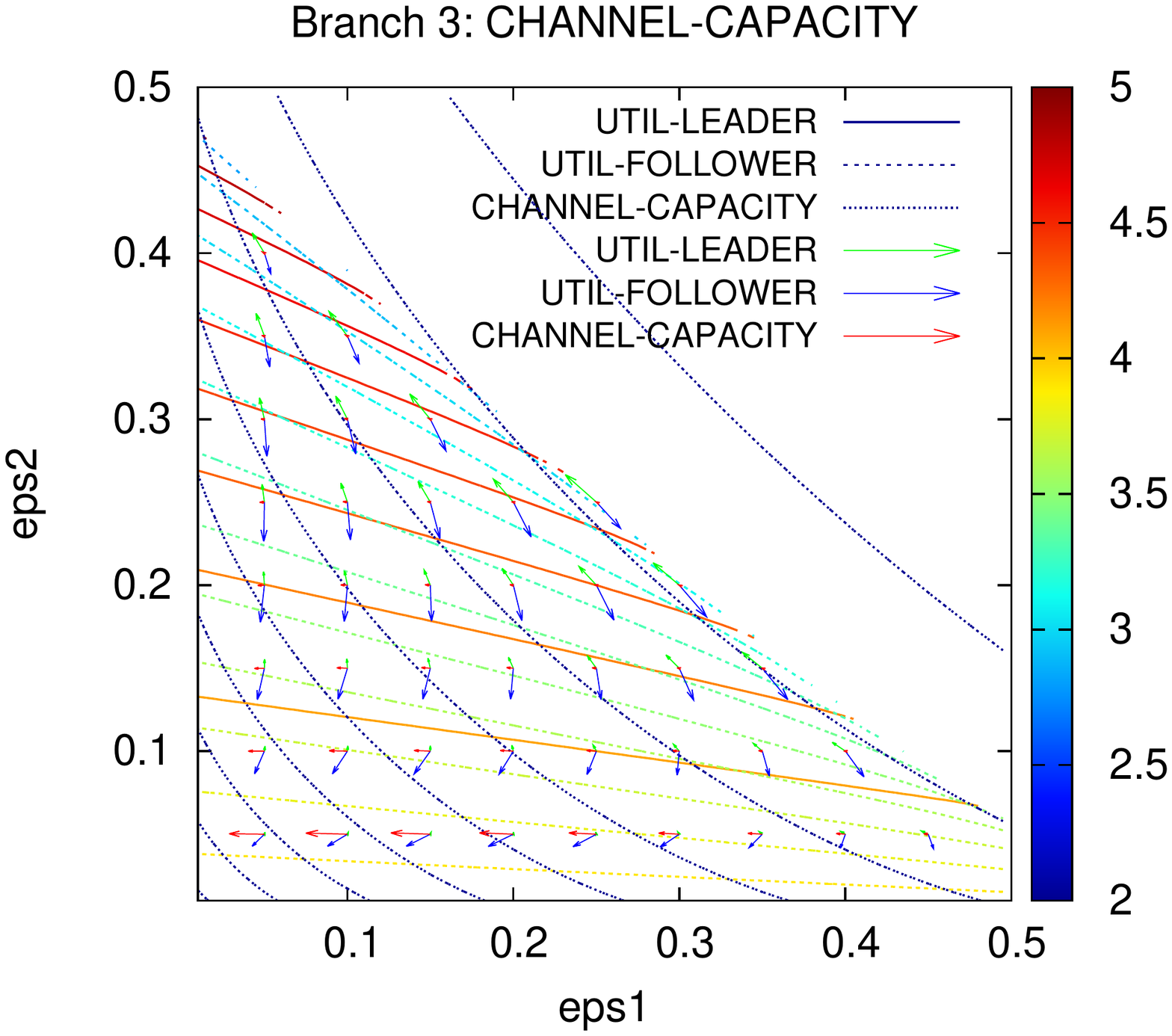}
  \end{tabular}
  \caption{\label{fig:StackelbergBranchIso} Isoclines of the expected
    utilities for the leader and the follower as well as the channel
    capacity in the asymmetric channel scenario, on the Stackelberg branch ({\bf A}) and on the branch just
    below it in the space of leader strategies
     ({\bf B}). The expected utility levels
    are color-coded and the channel capacity (dotted) increases
    towards the origin (lower-left corner). The gradient vectors show
    the corresponding directions of steepest ascent (wrt/ the Fisher
    metric). }
\end{figure}}

Fig.~\ref{fig:StackelbergBranchIso} {\bf A} shows the isoclines of
expected utility for both leader and follower, the isoclines of the
channel capacity, and the corresponding gradient fields, all on the
Stackelberg branch. We immediately see that the gradients are nowhere
collinear. So from Prop.~\ref{prop_nonalign} we know that for all
values of the game parameters, the channel noises can be
infinitesimally changed such that the capacity increases whereas the
expected utility (of either the leader or follower)
decreases.\footnote{By Prop.~\ref{prop_nonalign_generic} this behavior
  is rather generic and thus expected.} More interestingly, the
gradient field shows that everywhere $\grad [\mbox{Capacity}] \notin
Con(\{\grad(V)^1, \grad V^2\})$. Thus by Prop.~\ref{prop_allneg},
there must be directions such that both players are better off and yet
the channel capacity decreases when the game is moved in that
direction.  As before, while this is true locally, it does not mean
that both players prefer less information globally.
 
 Simply by redefining $f$ as the negative of the capacity as the
function of interest (which amounts to flipping the corresponding gradient
vectors in our figures), the same condition in Prop.~\ref{prop_allneg} can be used to identify regions of channel
parameter space where (there are directions in which we can move the game parameters in which)
both players prefer more capacity. Now $\grad
(-\mbox{Capacity}) \in Con(\{\grad(V)^1, \grad(V)^2\})$, except for a
small region in the upper left corner (containing for example the point $(\epsilon^1 = 0.05,
\epsilon^2 = 0.35)$). So we can immediately conclude that there are no
such directions, for all parameter values outside this region.

To illustrate the effect of our choice of which equilibrium branch to analyze,
Fig.~\ref{fig:StackelbergBranchIso} {\bf B} shows the corresponding
isoclines on the branch just below the Stackelberg branch.
Again, the gradient vectors are nowhere
aligned. However now $\grad[ \mbox{Capacity}]$ $\in$ $Con(\{\grad(V)^1, \grad
V^2\})$. Thus, we can conclude that on this branch there are no directions where both
players prefer a decrease in channel capacity, in contrast to the case on the Stackelberg branch.

Summarizing, even in rather simple
leader-follower games, there is a surprisingly complicated underlying geometric structure.
However  just like in the decision-theoretic scenario, our theorems provide 
general conditions on the gradient vectors in leader-follower games that  determine regions
of channel parameter space in which there is negative value of information for both
players. Moreover, these conditions are easily checked. In particular, they do
not require laboriously comparing different information
partitions to get non-numeric partial orders. Instead our conditions are based on cardinal
quantifications of how much the players value infinitesimal changes to the
information structure of the game, measured in units of utility per bits of information. 

This kind of analysis can be applied multiple times at once, to quantify how much the players
value different candidate changes to the information structure of the game. These quantifications 
are all  measured in the same units, of utility per bits of information. So we can use their values
to evaluate marginal rates of substitution of various kinds of changes to the game's information structure. 
More generally, we can apply our analysis to quantify how much the players
value different candidate infinitesimal changes to \emph{any} aspects of the game specification,
e.g., to parameters in utility functions. This allows us to evaluate marginal rates of substitution
of all aspects of the game specification. This is the subject of future research.

\subsection{Illustrative examples}

To illustrate the generality of
our framework, in this section we present additional examples, emphasizing the
implications of our results in Sec.~\ref{sec:properties} for several different
economic scenarios.

\label{sec:more_examples}

\begin{example}
Consider a variant of the well-known scenario where a decision by an individual 
is a costly signal to a prospective employer of their capability as employees. 
In this variant there is a car repair chain that wants to hire a new mechanic. (So we have two players.)
There is an applicant for the position who has some pre-existing ability at car repair.
That ability is their type; it is determined by a prior distribution that neither
player can affect. The repair chain cannot directly observe the applicant's ability.
So instead, they will give the applicant a written test of their
knowledge of cars. The repair chain will then use the outcome of that test
to decide whether to offer the job to the applicant and if so, at what salary. 
(The idea is that giving too low a salary to a new mechanic will raise the odds that after 
being trained by the repair chain that new mechanic would simply leave for a different
repair chain, at a cost to the repair chain.)

The applicant will study before they take the test. The harder they study, the
greater the cost to them. There is also a conditional distribution of how they
do on the test given their ability to repair cars and how hard they study. The fact that that distribution is
not a single-valued map means it is a noisy information channel, from the
studying decision of the applicant to the outcome of the test. 

Suppose that the repair chain feels frustrated by the fact that
the test gives them a noisy signal of the applicant's ability to
repair cars. Knowing this, a test-design company approaches the repair
chain, and offers to sell them a new test that has a double-your-money
back guarantee that it is less noisy than the current test (for some
functional of how ``noisy" a test is that the test company and repair
chain both use).  Thinking they will get double their money back if
they buy the new test but have lower expected utility, the repair
chain buys the test.{\footnote{Formally, in this example we must
    assume that the applicant knows nothing about this option that the
    repair chain has to buy a new test before examining the
    applicant. Rather the applicant is simply informed about the
    conditional distributions specifying the accuracy of the test ---
    whatever they are ---- before the applicant considers the test. If
    instead the applicant knew that the repair chain has the option to
    purchase the new test, it would mean that we have to consider an
    expanded version of the original game, in which the applicant must
    predict whether the repair chain purchases the new test, etc.}}

Our results show that generically, there are ways that the repair
chain will have lower expected utility with the new test --- but not
be able to invoke the guarantee to get any money back from the
test-design company, since the new test \emph{is} less noisy than the
old one. That is, there are directions $\delta \bftheta$ in the
parameters describing how the test is designed, such that the test is
made more accurate, but less useful for the repair chain, i.e. it has
a negative value of information.
\label{ex:noisy_signalling}
\end{example}

\begin{figure}
  \begin{tabular}{ll}
    {\bf A} & {\bf B} \\
    \begin{tikzpicture}
      \tikzstyle{nature}=[draw=blue!60,fill=blue!20,circle,minimum size=5mm]
      \tikzstyle{player}=[draw=red!60,fill=red!20,rectangle,minimum size=8mm]
      \tikzstyle{utility}=[draw=green!60,fill=green!20,diamond,minimum size=8mm]
      \tikzstyle{standard}=[shorten >=1pt,>=stealth',semithick]
      \node (nature) at ( 2,4) [nature] {$T$};
      \node (player1) at ( 2,2) [player] {$A_1$};
      \node (channel1) at ( 2,0) [nature] {S};
      \node (player2) at ( 2,-2) [player] {$A_2$};
      \node (utility2) at (4,1) [utility] {$U_2$};
      \node (utility1) at (0,1) [utility] {$U_1$};
      \draw [->] (nature) edge[standard] (player1);
      \draw [->] (player1) edge[standard] (channel1);
      \draw [->] (channel1) edge[standard] (player2);
      \draw [->] (player1) edge[standard] (utility1);
      \draw [->] (player2) edge[standard] (utility2);
      \draw [->] (nature) edge[standard] (utility1);
      \draw [->] (nature) edge[standard] (utility2);
      \draw [->] (player2) edge[standard] (utility1);
    \end{tikzpicture}
    &
    \begin{tikzpicture}
      \tikzstyle{nature}=[draw=blue!60,fill=blue!20,circle,minimum size=5mm]
      \tikzstyle{player}=[draw=red!60,fill=red!20,rectangle,minimum size=8mm]
      \tikzstyle{utility}=[draw=green!60,fill=green!20,diamond,minimum size=8mm]
      \tikzstyle{standard}=[shorten >=1pt,>=stealth',semithick]
      \node (player1) at ( 4,2) [player] {$A_1$};
      \node (channel1) at ( 4,0) [nature] {$S$};
      \node (player2) at ( 4,-2) [player] {$A_2$};
      \node (utility2) at (6,1) [utility] {$U_2$};
      \node (utility1) at (2,1) [utility] {$U_1$};
      \draw [->] (player1) edge[standard] (channel1);
      \draw [->] (channel1) edge[standard] (player2);
      \draw [->] (player1) edge[standard] (utility1);
      \draw [->] (player1) edge[standard] (utility2);
      \draw [->] (player2) edge[standard] (utility2);
      \draw [->] (player2) edge[standard] (utility1);
    \end{tikzpicture}
  \end{tabular}
  \caption{\label{fig:noisy_signalling} MAID for the noisy signalling
    game described in example~\ref{ex:noisy_signalling} (panel {\bf
      A}) and the leader-follower game from
    section~\ref{sec:two_player} (panel {\bf B}).}
\end{figure}
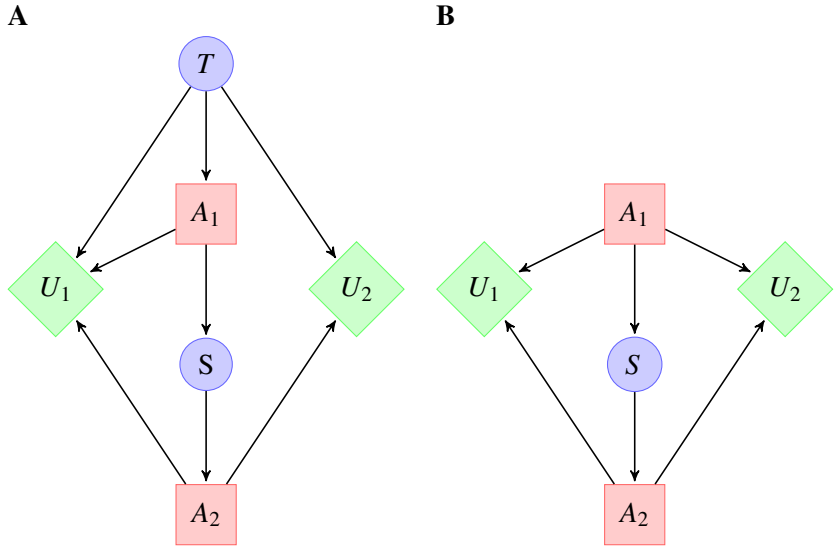

The MAID corresponding to this noisy signaling game is shown in
Fig.~\ref{fig:noisy_signalling} {\bf A}. For comparison, panel {\bf B}
reproduces the MAID for the leader-follower game that was studied in
section~\ref{sec:two_player}. In the noisy signaling game
there is an additional Nature node $T$ which player 1 (the applicant) can observe, but that
player 2 (the car repair chain) cannot observe. Another difference is that now the utility
of both players depends in part on the outcome of Nature's choice.  Since the applicant
incurs a cost if she studies, her utility directly depends on her move (study
hard or party). However the utility of the repair chain does not depend \emph{directly}
on the move of the applicant.

We now present a variant of Braess' paradox~\cite{Braess1968} which
provides a particularly striking example of how extra information can
simultaneously hurt all players of a game, even in the case of many
players. Braess' paradox is a ``paradox" that arises in congestion games over
networks. Although it has arisen in the real world in quite complicated
transportation networks, it can
be illustrated with a very simple network. We illustrate in this next example,
before presenting our variant of Braess' paradox.

\begin{example}
  Consider a scenario where there are 4000 players, all of whom must
  commute to work at the same time, from a node ``Start" on a network
  to a node ``End". Say that there are a total of four roadways
  between Start and End that all the players know exist.  (See
  Fig.~\ref{fig:braess}.)

  \showfigure{
    \begin{figure}[h]
      \includegraphics[width=\textwidth]{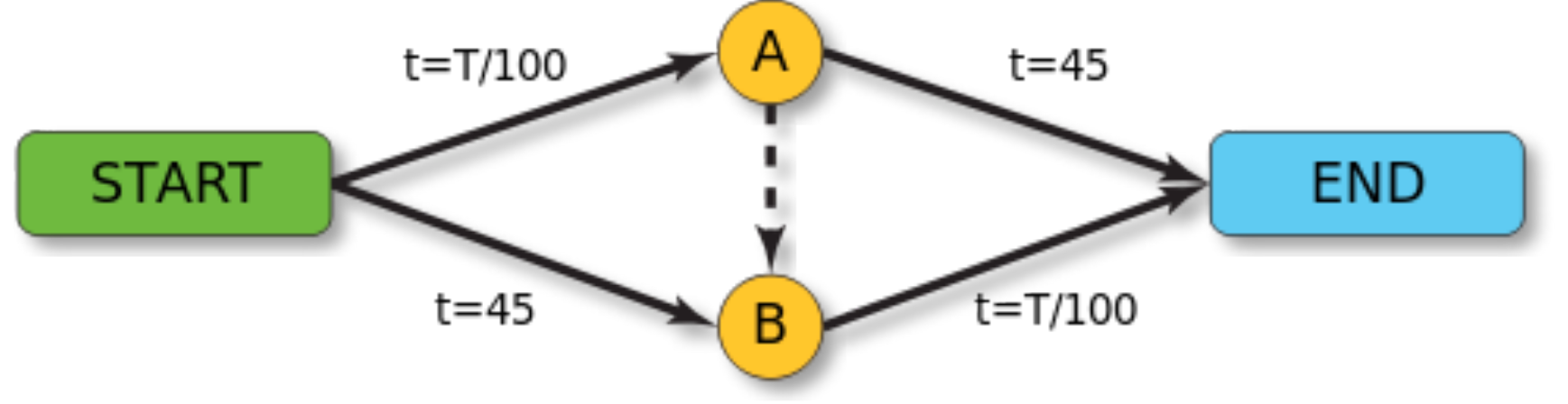}
      \noindent \caption{Network exhibiting Braess' paradox. The symbols are explained in the text. }
      \label{fig:braess}
    \end{figure}}
  
  The move of each player is a choice of what route they will follow
  to get to work. There are two choices they can make. The first is
  the route Start-A-End, and the second is the route Start-B-End.  The
  number of minutes it takes each player to follow their chosen route
  is illustrated in Fig.~\ref{fig:braess}.  In that figure, $t$
  indicates the amount of time it takes a player to cross the
  associated road. In addition, $T$ refer to the total number of
  players (out of 4000) who follow the associated road. So on those
  two roads where $t = T / 100$, the greater the amount of traffic
  $t$, the slower the players can drive, and so the greater the amount
  of time $T$ it takes to cross the road. In contrast, it takes 45
  minutes to traverse the other two roads, regardless of the amount of
  traffic on the roads.

  The paradox arises, if a new highway from A to B (dashed) is opened
  which takes only 4 minutes to traverse. Now, each player has a third
  option, namely to take the route Start-A-B-End. So we have a new game.
  
  It is straightforward to verify that in all Nash equilibria of the
original game, i.e. the game without the new highway, exactly half of the
players choose Start-A-End, and half choose Start-B-End. The total
cost (i.e., negative utility) to all players is a 65 minute commute.
When the new highway is opened, every player will choose Start-A-B-End.
This increases the travel time of all players to 84 minutes. Thus,
the expected cost for \emph{all} players rises when they have more
possible actions.

Now consider a variant of Braess' paradox where
both games are changed slightly, so that they only differ in their information structures. In this variant,
the new highway exists in both games, but is closed due to
construction with prior probability of .1. We assume that the cost of 4 minutes
is incurred by a player if they try to go down the new highway both
if they are successful and get through to B, or if they are blocked
by the new construction and have to return to A. In addition,
in both games all players are informed
about whether the new highway is open, but via a noisy signal
(e.g. via a newspaper article from several days before saying when the
new highway is schedule to open). The difference between the games is
the amount of noise in that signal.

In the new first game, the signal is completely noisy, 
providing no information at all about whether the new highway is open.
So the players have to decide their strategy based purely on the prior probability of .1
that the new highway is open. It is straightforward to verify that in this case,
the Nash equilibrium is for half of the players to choose Start-A-End,
half to choose Start-B-End, and none to try to go down the new highway.
The resultant travel time per player is 65 minutes.{\footnote{In particular, if
any of the players in this strategy profile who choose Start-A-End were to
change their choice to try to go down the new highway, their new expected
travel time would be $.1(24 + 20.01) + .9(24 + 45) = 66.5$, which is
greater than 65, their travel time if they stick with their original strategy.}}

In contrast, in the new version of the second game, the signal is
noise-free. So with probability .1 the new highway is open,
all players know that, and therefore
all take the new highway, for a total travel time of 84 minutes. With
probability .9 the new highway is closed, all players know that, and
therefore none try to take the new highway, for a total travel time of
65 minutes. So the expected travel time per player in the new second game is
$.1(84) + .9(65)$.

So the extra information that is available in the new second game, but not in the new first game,
hurts all players.

  \label{ex:braess}
\end{example}

Note that in contrast to the leader-follower games and noisy signaling games
analyzed above, in our variant of Braess' paradox  the extra information concerns a move of Nature, not the
move of some other player. Nevertheless, our results on negative value
of information still apply. In particular, for any
particular player $i$ in that game, and any precise choice for
information-theoretic function $f$, Prop.~\ref{prop_nonalign} tells us that
there are other directions $\delta
\vec{\theta}$ in which we infinitesimally reduce the noise
in the signal about the state of the highway so that information increases \emph{and
  expected utility for player $i$ goes up}. 
Prop.~\ref{prop_nonalign_generic} then tell us that 
this property is generic, i.e., it is true for
almost all utility functions that differ only slightly from the ones
in Ex.~\ref{ex:braess}.

This may have important real-world applications. Computer routing
networks (e.g., networks that route communication packets, jobs, etc.), are
typically run in a distributed fashion where each router adaptively
modifies its behavior to try to optimize its own ``utility function"
based on noisy signals it receives concerning the state of the overall
network. In addition to arising in human transportation traffic
networks, Braess' paradox often arises in such computer routing
networks~\cite{Rough02}. This has led to a large body of work on how to redesign the
adaptive algorithms used by the routers to avoid the paradox and
associated loss of expected
utility ~\cite{wotu02a,LeyBrown09}.  Our
analysis suggests that it should be possible to avoid Braess' paradox
--- and indeed to increase the expected utility --- without
redesigning the routing algorithms, but instead changing the signals
the routers receive concerning the state of the network.

\section{Future work}

In this paper we primarily used our geometric framework to analyze the
relationship between changes in information and associated changes in
expected utility.  However the framework is far more broadly
applicable. It can be used to analyze the relationship between
expected utilities and \emph{any} function $f(\vec{\theta})$ that depends on the parameters
specifying the game. $f$ is not restricted to be an information-theoretic function.

As a particularly simple illustration of this breadth of the applicability of our framework, we can use
it to analyze the relationship between expected utility and a
function $f(\vec{\theta})$ that simply returns one of the components
of $\vec{\theta}$.  This analysis reveals what might be called scenarios with
``negative value of utility":

\begin{example}
  This example concerns a simultaneous move game of two players, who
  have two possible moves each. The bimatrix of the game is
  as follows:

  \begin{center}
    \begin{tabular}{ l  | c | r }
      & \;\;\; \textbf{L} & \textbf{R}  \;\;\; \\ \hline
      \textbf{T} \; & (1,4) & (4-$\theta$,1)\\ \hline
      \textbf{B}\; & (2,2) & (3,3) 
    \end{tabular}
  \end{center}
  where $\theta \ge 0$ is a parameter adjusting the utility of the row
  player.
  
  With $\theta < 1$ this game has no pure strategy Nash
  equilibrium. The unique Nash equilibrium of this game is a mixed
  equilibrium with the row player playing top with $p_T=\frac{1}{4}$
  and the column player playing left with
  $p_L=\frac{1-\theta}{2-\theta} = 1 - \frac{1}{2 - \theta}$. Since
  the column player is indifferent between $L$ and $R$ at the
  equilibrium, the expected utility $V_{col}$ is $\frac{5}{2}$
  independent of $\theta$. While the row player has an expected
  utility of $V_{row} = \frac{5 - 2 \theta}{2 - \theta}$.

  Now, $\frac{\partial V_{row}}{\partial \theta} = \frac{1}{(2 -
    \theta)^2}$ is strictly positive which means that the row player
  prefers $\theta$ too increase. At the same time, this reduces her
  utility for the outcome $T, R$ and thus, one could say that the row
  player has a {\em negative value of utility}.{\footnote{As an aside,
      this phenomenon can be exploited by an external party who can
      enter a publicly visible binding contract with Row under which
      Row must pay the tax to the external party, to the benefit of
      both Row and that external party.  This is an example of an
      external party ``mining'' a
      game~\cite{bono_wolpert_game_mining_2014}.}}

  \label{ex:neg_val_ut}
\end{example}

As Ex.~\ref{ex:neg_val_ut} illustrates, 
our framework allows us to quantify the value of any infinitesimal change in
the change in any function $f(\vec{\theta})$ induced by
by an infinitesimal change in the game parameter. Indeed, we can evaluate such a value even if
the induced change in $f(\vec{\theta})$
is indirect, arising via the effect of the change in $\vec{\theta}$ on the
player strategy profile that is mediated by the equilibrium concept (as in the
leader-follower game from Sec.~\ref{sec:two_player} and as in
Ex.~\ref{ex:braess}).  

This breadth of quantities whose ``value" can be evaluated using our framework allows us to analyze the
trade-offs, inherent in a game, among multiple changes to the parameters specifying
that game and/or functions of those parameters. These trade-offs can be used to calculate
marginal rates of substitution, in which we compare the values assigned by a single
player to different changes in the game parameter. Assuming transferrable utility,
they can also be used to compare the values assigned by different
players to the same change in the game parameter. Under that assumption they can even be used to compare the values
assigned by different players to different changes in the game parameter. For example, 
with our framework we can
quantify the relationship between the value to player $i$ of changes to the information capacity of one information
channel in a game, and the value  to a separate player $j \ne i$ of 
changes to $i$'s utility function in that game. Similarly, by
considering the QRE rationality parameter $\beta$ as a component of the game parameter $\vec{\theta}$, we can study the
value of rationality and its relations with values of other quantities. This
allows us to do things like quantify the relationship between extra rationality by player $i$ and extra information
available to that player. Colloquially, ``$i$'s knowing this much more about certain utility-relevant quantities before
they make their move is equivalent to their being this much smarter when they make that move". 
Future work involves making a detailed investigation of these kinds of trade-offs.

Another issue we intend to investigate in future work is ``second-order effects". 
Changing the parameter of a game $\vec{\theta}$ infinitesimally
can affect the value of \emph{every} function $g: \vec{\theta} \in \Theta \rightarrow \R$, not just functions like
expected utilities, mutual information, etc.
In particular, this is true if $g$ is the differential value of 
some $f(\vec{\theta})$ evaluated at $\vec{\theta}'$. 
Changing $\vec{\theta}$ will not just change the expected utility of player $i$ and
what $f(\vec{\theta})$ is; it will also change the differential value to player $i$ of changes to $f$.
As an example, changing the conditional distribution specifying one
information channel in a game will in general change the differential value
of a \emph{different} information channel. 
In future work we hope to investigate these second order effects and whether they
depend on ``second-order" properties of the geometric structure of the game, like
the Ricci curvature tensor of the Fisher metric.

Since in general games have multiple
equilibria, it is often impossible to infinitesimally change the game parameter in 
way that is Pareto-optimal simultaneously for all equilibria. In other future work, we intend to investigate
the generalizations of Prop.~\ref{prop:pareto_neg_val_f} that determine when such changes 
are possible.

Finally, our framework provides important new capabilities for
policy making or mechanism design, by providing guidance to a regulator external to
a game on how to modify the components of the game parameter vector
that are under their control. As discussed above, in games
with multiple equilibria --- arguably the majority of real-world games of
interest to a regulator --- it makes sense for the regulator to determine which
equilibrium the players have adopted for a current value of $\vec{\theta}$ simply by observing how the
players are behaving. Using our framework, it may be possible for that regulator
to gradually change  $\vec{\theta}$ from that starting value, to move the equilibrium
down the branch that has the current equilibrium, to where it intersects with
another branch, and then guide the equilibrium along that second branch,
back to the starting  $\vec{\theta}$.
In this way the regulator may be able to gradually change player behavior 
to go from a current equilibrium of a game specified by  $\vec{\theta}$, to a different, Pareto-superior equilibrium
for that same  $\vec{\theta}$. In some situations, the regulator should even
be able to design that trajectory through the space of game parameter vectors so
that each infinitesimal change to  $\vec{\theta}$ along that trajectory is Pareto-improving. (See \cite{woha12} for an example of
this kind of approach.)

\section{Conclusion}

In this paper we introduce a new framework to study the value of information
in arbitrary noncooperative games.
Our starting point is the well-established concept of marginal utility of a good in a decision problem.
We present a very natural way to generalize this, to marginal utility to a player in a noncooperative game
of an arbitrary function of a parameter specifying that game. Interestingly, this generalization
forces us to introduce a metric over the space of game parameters. In this
way we show that geometry is intrinsic to noncooperative game theory,
with each game specifying its own associated Riemannian manifold. 

In parallel with this analysis, we consider the issue of how best to quantify 
economically relevant aspects of the information
structure in that game. We argue that
mutual information is a natural way to do this, using very simple considerations grounded in economics.
As we discuss, such a (cardinal) quantification of information also has several advantages over the partial
orders commonly used in the past to investigate the role of information in games.

We then combine our two analyses into a unified framework. We do this by taking the ``function of a parameter
specifying a game" from the first analysis to be the mutual information
that was motivated in our second analysis. This combination 
of our two analyses reveals how a game's geometry governs the relation between
changes to its information structure and changes to the expected
utilities of the players.

We then use our framework to derive general
conditions for the existence of negative value of information. In
particular, we show that for \emph{almost any} game, there are 
changes to the information structure of
the game that both increase the information available to any particular player
in that game and hurt that player.  We then extend our analysis to
characterize the set of games where there are such changes that simultaneously 
increase the information available to all players, while hurting all players.
We illustrate our framework with computational analyses of a single-player decision scenario as well as a two
player leader-follower game.

Finally, we note that our framework can by applied to analyze the effects
of arbitrary changes to a game, not just changes to its information structure.
As a particularly simple example, we construct a game that has
``negative value of utility'', in which the expected value of a player's utility $u$ \emph{increases}
when we change the game by applying a monotonically \emph{decreasing} transformation to $u$. More generally,
the breadth of applicability of our framework allows us to analyze marginal
rates of substitution of different aspects of an information structure, of the utility
functions of the players, or of any other parameters specifying the game the players
are engaged in.

{\subsection*{Acknowledgments} NB acknowledges support by the Klaus
  Tschira Stiftung. The research of JJ was supported by the ERC
  Advanced Investigator Grant FP7-267087. DHW acknowledges support of
  the Santa Fe Institute.}

\begin{appendix}
\section{Review of conic hulls}
\label{app:ConicHull}

In our analysis below
we shall work in some tangent space of a parameter manifold that is equipped with the Fisher metric $\langle .,. \rangle$. Now, whenever we have  a scalar product on such a (finite dimensional) vector space, we can perform a linear transformation of that vector space to turn that scalar product into the standard Euclidean one. Thus, in 
our analysis below, we shall essentially be doing elementary Euclidean geometry, just in a different
coordinate system.{\footnote{The nonlinear nature of the Fisher or any other Riemannian metric only comes into play when we look at the tangent spaces of several points simultaneously. What linear coordinate transformation turns the Fisher metric into the Euclidean one will  depend on the particular tangent space in which we are working for a particular $\vec{\theta}$, and in general, there will be no such transformation 
that works for all tangent spaces simultaneously.}} This will be reflected in the fact that simple graphical pictures and intuition
can be used to understand many of our results.

We start by reviewing some simple linear algebra in Euclidean space, that is, $\R^d$ equipped with the Euclidean scalar product $\langle .,. \rangle$. 
For a nonzero vector $\vec{v} \in \R^d$, the hyperplane $H_0(\vec{v})=\{ \vec{w}:  \langle \vec{v}, \vec{w}\rangle =0\}$ separates $\R^d$ into two halfspaces $H_{+}(\vec{v}) \; (H_-(\vec{v}) \; resp.) \equiv \{ \vec{w}  :  \langle \vec{v},\vec{w}\rangle >0 \; (<0 \; resp.)\}$, and for two vectors $\vec{v}_1,\vec{v}_2$ that are not positively collinear, the
associated halfspaces overlap, e.g., 
  \begin{equation}
    \label{eq:la1}
H_+(\vec{v}_1)\cap H_-(\vec{v}_2) \neq \emptyset. 
  \end{equation} 

When we have several nonzero vectors $\vec{v}_i$,   their \textbf{conic hull}~\cite{bova03} is defined as
\begin{equation}
\label{eq:la2}
Con(\{\vec{v}_i\}) \equiv \left\{\sum_i \alpha^i \vec{v}_i : \alpha_i \ge 0 \forall i \right\}
\end{equation}
Note that $C=Con(\{\vec{v}_i\})$ is a cone since whenever $\vec{v}\in C$, it follows that $k \vec{v} \in C$ for $k > 0$. Note also
that this cone is convex because whenever $\vec{v},\vec{w} \in C$, then $\lambda \vec{v} +(1-\lambda) \vec{w} \in C$ for $0\le \lambda \le 1$. 

A cone $C$ is called \textbf{pointed} if it does not contain any
bi-infinite straight line. The following will be used in the
sequel\shortversion{\footnote{For the sake of space we omit the proof
    of this and the following basic properties.}}:
\begin{lemma}\label{lemma:conic}
If the vectors $\vec{v}_i$ are linearly independent, then $Con(\{\vec{v}_i\})$ is pointed. 
\label{lemma:1}
\end{lemma}
\fullversion{
\begin{proof}
If  $Con(\{\vec{v}_i\})$ were not pointed, there would exist some $\vec{v} \neq 0$ with both  $\vec{v}$ and $-\vec{v}$ in $Con(\{\vec{v}_i\})$. Thus, we have both $\vec{v}=\sum_i \lambda_i \vec{v}_i$ and $-\vec{v}=\sum_i \mu_i \vec{v}_i$ with $\lambda_i \ge 0, \mu_i \ge 0$ for all $i$, and not all of them simultaneously 0. But then $\sum_i (\lambda_i +\mu_i) \vec{v}_i=0$ which is a nontrivial linear relation between the $\vec{v}_i$ violating our assumption that they be linearly independent. 
\end{proof}
}
\noindent Of course, the converse of Lemma~\ref{lemma:1} does not hold in general, because for all $\vec{v}\in Con(\{\vec{v}_i\})$, $Con(\{\vec{v}_i\})=Con(\{\vec{v}_i\}
\cup \vec{v})$.

We will also need to use the definition that the \textbf{dual} to a conic hull $Con(\{\vec{v}_i\})$ is
\begin{equation}
  \label{eq:la3}
Con(\{\vec{v}_i\})^\bot \equiv \bigcap_i H_-(\vec{v}_i).  
\end{equation}
Equivalently, 
\begin{equation}
  \label{eq:la4}
  Con(\{\vec{v}_i\})^\bot =\{ \vec{v}: \langle \vec{v}, \vec{v}_i \rangle <0 \text{ for all } i\}.
\end{equation}

The following elementary property  relating conic hulls and pointedness is used in the sequel.
\begin{lemma}
$Con(\{\vec{v}_i\})$ is not pointed $\Rightarrow$ $[Con(\{\vec{v}_i\})]^\bot = \varnothing$ $\Rightarrow$ the vectors $\vec{v}_i$ are
linearly dependent.
\label{lemma:conic_props}
\end{lemma}
\fullversion{
\begin{proof}
We prove the implications in turn.

If $Con(\{\vec{v}_i\})$ is not pointed then there is a non-null vector $\vec{u}$ such that both $\sum_i a_i \vec{v}_i = \vec{u}$ and $\sum b_i \vec{v}_i = -\vec{u}$
where no $a_i$ or $b_i$ is less than zero. Adding these two expressions, there is some set of values $c_i$
none of which are less than zero and at least one of which is nonzero such that $\sum_i c_i \vec{v}_i = 0$.

Now hypothesize that $[Con(\{\vec{v}_i\})]^\bot \ne \varnothing$. This would mean that there is some vector $\vec{w}$ such that
$\langle \vec{w}, \vec{v}_i\rangle < 0$ for all $i$. However by definition of the $c_i$, for any vector $\vec{w}$,
\ba
\sum_i c_i \langle \vec{w}, \vec{v}_i \rangle &=& 0 \nonumber
\ea
This contradicts the fact that no $c_i < 0$ and at least one $c_i > 0$. So our hypothesis cannot hold. This
establishes the first implication.

To establish the second implication, note that if $N > d$, then the $\vec{v}_i$ are linearly dependent always, so we only have to consider
the cases $N = d$ and $N < d$. If $N = d$ and the $\vec{v}_i$ are linearly independent, then the $N \times N$ matrix $C$
whose rows are the $\vec{v}_i$ is non-singular. Therefore for \emph{any} vector $\vec{b}$, we can solve $C_j^k a^j = b^k$ for the vector $\vec{a}$. 
So in particular, we can do this if all the components of $\vec{b}$ are less than
zero. For such a $\vec{b}$ and associated vector $\vec{a}$, the dot product between $\vec{a}$ and $\vec{v}_i$ is less than zero for
all $i$, i.e., $\sum_k a^k [v_i]^k < 0$ for all $i$. Moreover, since we restrict attention to Riemannian metrics, the 
inner product between two contravariant
vectors is just the dot product: $\langle \vec{a}, \vec{b} \rangle = \sum_k \vec{a}^k \vec{b}^k$. Hence 
$\vec{a} \in [Con(\{\vec{v}_i\})]^\bot$, i.e., $[Con(\{\vec{v}_i\})]^\bot$ is non-empty.

The remaining case to consider is where $N < d$. Create a matrix $C$ whose first $N$ rows are the vectors
$\vec{v}_i$, and the remaining $d - N$ rows are linearly independent both from one another
and from the first $N$ rows. If the $\vec{v}_i$ are also independent of one another, then $C$ is non-singular. So again,
we can solve for a vector $\vec{a}$ such that $C_j^k a^j = b^k$ where the first $N$ entries of $\vec{b}$ are all
less than zero. And so again, we conclude that this $\vec{a} \in [Con(\{\vec{v}_i\})]^\bot$, i.e., $[Con(\{\vec{v}_i\})]^\bot$ is non-empty.
\end{proof}
}

For any single vector $\vec{v}$, $Con(\vec{v})^\bot =H_-(\vec{v})$. In general, when we enlarge the set of vectors $\vec{v}_i$, $Con(\{\vec{v}_i\})$ gets larger whereas the dual conic hull $Con(\{\vec{v}_i\})^\bot$ becomes smaller. More precisely, 
\begin{lemma}
  Let $C,C_1,C_2$ be nonempty convex cones whose duals are nonempty. Then
  \begin{equation}
    \label{eq:la5}
    C^{\bot \bot}=C
  \end{equation}
and
\begin{equation}
  \label{eq:la6}
  C_1 \subset C_2 \quad \Leftrightarrow \quad C_2^\bot \subset C_1^\bot.
\end{equation}
\label{lemma:conic_hull_duals}
\end{lemma}
\noindent The requirement that the duals be non-empty is crucial in this result. For example, 
if the second part of Lemma~\ref{lemma:conic_hull_duals} held for empty $C^\bot$, then
we would have $C^\bot = \varnothing \Rightarrow C' \subset C$ for any
conic hull $C' \subseteq \R^n$. The only way this could be would be if $C =\R^n$. However 
for any vector $\vec{v} \in \R^n$ with $n > 2$, $[Con(\{\vec{v}, -\vec{v}\})]^\bot = \varnothing$, even though $[Con(\{\vec{v}, -\vec{v}\})]
\ne\R^n$.

\section{Partial derivatives of QREs of MAIDs with respect to game parameters}
\label{app:partial_derivative}

For computations involving the partial derivatives of the
players strategies at a QRE (branch) it can help to explicitly introduce the
normalization constants as an auxiliary variable. The QRE condition
from \equ{QRE} is then replaced by the following conditions
\begin{eqnarray*}
\sigma_i(a_v | x_{pa(v)}; \beta_i,\vec{\theta}) - \frac{e^{\beta_i \E(u_i | a_v, x_{pa(v)}; \vec{\beta}, \vec{\theta})}} {Z_i(x_{pa(v)}; \beta_i,\vec{\theta})} & = & 0 \\
Z_i(x_{pa(v)}; \beta_i,\vec{\theta}) - \sum_{a \in \set{X}_v} {e^{\beta_i \E(u_i | a_v, x_{pa(v)}; \vec{\beta}, \vec{\theta})}} & = & 0
\end{eqnarray*}
for all players $i$, decision nodes $v \in \set{D}_i$ and all states
$a_v \in \set{X}_v, x_v \in \cart_{u \in Pa(v)} \set{X}_u$. (Here and throughout
this section, subscripts on $\sigma$, $Z$, etc. should not be understood as 
specifications of coordinates as in the Einstein summation convention.)

Overall, this gives rise to a total of $M$ equations for $M$ unknown
quantities $\sigma_i(a_v|x_{pa(v)})$, $Z_i(x_{pa(v)})$. 
Using a vector valued
function $\vec{f}$ we can abbreviate the above by the following
equation:
\be
\vec{f}(\vec{\sigma}_{\vec{\beta},\vec{\theta}}, \vec{Z}_{\vec{\beta},\vec{\theta}}, \vec{\beta},\vec{\theta}) = \vec{0}
\label{eq:QREformal}
\ee where $\vec{\sigma}_{\vec{\beta},\vec{\theta}}$ is a vector of all
strategies 
\[ \{\sigma_i(a_v \mid x_v; \beta_i,\vec{\theta}) : i = 1,\ldots,n , v \in \set{D}_i, a_v \in \set{X}_v, x_v \in \cart_{u \in
  Pa(v)} \set{X}_u\}, \]

\noindent $\vec{Z}_{\vec{\beta},\vec{\theta}}$ collects
all normalization constants, and $\vec{0}$ is the 
vector of all 0's. Note that in general, even once the distributions at all 
decision nodes have been fixed, the distributions at chance nodes 
affect the value of $ \E(u_i \mid a_v, x_{pa(v)}; \vec{\beta}, \vec{\theta})$.
Therefore they affect the value of the function $\vec{f}$. This
is why $\vec{f}$ can depend explicitly on $\vec{\theta}$, as well as depend directly on
$\vec{\beta}$.

The (vector-valued) partial derivative of the position of the QRE in
$(\vec{\sigma}_{\vec{\theta}}, \vec{Z}_{\vec{\theta}})$ with respect
to $\vec{\theta}$ is then given by implicit differentiation of \equ{QREformal}
:
  \be
\left[ \begin{array}{c} \frac{\partial \vec{\sigma}_{\vec{\theta}}}{\partial \vec{\theta}} \\
    \frac{\partial \vec{Z}_{\vec{\theta}}}{\partial
      \vec{\theta}} \end{array} \right] 
      = - \left[ \frac{\partial
    \vec{f}}{\partial \vec{\sigma}_{\vec{\theta}}} \quad \frac{\partial
    \vec{f}}{\partial \vec{Z}_{\vec{\theta}}} \right]^{-1}
\frac{\partial \vec{f}}{\partial \vec{\theta}}
\label{eq:QREpartial}
\ee where the dependence on $\vec{\beta}$ is hidden for
clarity, all partial derivatives are evaluated at the QRE, and we
assume that the matrix $\left[ \frac{\partial \vec{f}}{\partial
    \vec{\sigma}_{\vec{\theta}}} \quad \frac{\partial
    \vec{f}}{\partial \vec{Z}_{\vec{\theta}}} \right]$ is invertible
at the point $\vec{\theta}$ at which we are evaluating the partial derivatives.

These equations give the partial derivatives of the mixed strategy
profile. They apply to any MAID, and allow us to write the partial
derivatives of other quantities of interest. In particular, the
partial derivative of the expected utility of any player $i$ is
\begin{equation}
  \frac{\partial V_i}{\partial \vec{\theta}} = \sum_{x \in \set{X}_{\set{V}}} u_i(x) \frac{\partial p(x ; \vec{\theta})} {\partial \vec{\theta}}
  = \sum_{x \in \set{X}_{\set{V}}} u_i(x)  \sum_{v \in \set{V}} \frac{\partial p(x_v \mid x_{pa(v)}; \vec{\theta})} {\partial \vec{\theta}}  \prod_{v' \ne v} p(x_{v'} \mid x_{pa(v')}; \vec{\theta})
  \label{eq:partial_exp_ut}
\end{equation}
where each term $\frac{\partial p(x_{v} \mid x_{pa(v)};
  \vec{\theta})}{\partial \vec{\theta}}$ is given by the appropriate
component of Eq.~\eqref{eq:QREpartial} if $v$ is a decision
node. (For the other, chance nodes, $\frac{\partial p(x_v \mid
  x_{pa(v)}; \vec{\theta})}{\partial \vec{\theta}}$ can be calculated directly).
Similarly, the partial derivatives of other functions of interest such
as mutual informations between certain nodes of the MAID can be
calculated from Eq.~\eqref{eq:QREpartial}. 

Evaluating those derivatives and the additional
ones needed for the Fisher metric by hand can be very tedious, even
for small games. Here, we used automatic
differentiation \cite{AutoDiff} to obtain numerical results for certain parameter
settings and equilibrium branches. Note that
automatic differentiation is not a numerical approximation, like 
finite differences or the adjoint method. Rather it uses the chain rule to evaluate the derivative
alongside the value of the function.

\end{appendix}

\bibliographystyle{econometrica}
\bibliography{infogames-literature}
\end{document}